\documentclass[11pt]{article}
\usepackage[utf8]{inputenc}
\usepackage{lmodern}
\usepackage{amsthm, amsmath, amssymb, amsfonts, url, booktabs, tikz, setspace, enumerate}
\usepackage{fancyhdr}
\usepackage{mathtools}
\usepackage{graphicx,algorithmic,algorithm}
\usepackage[letterpaper, margin=1in]{geometry}
\usepackage{tikz-cd}
\usetikzlibrary{shapes}
\usepackage{hyperref}
\usepackage[nameinlink,capitalize]{cleveref}
\usepackage{authblk}
\usepackage{caption}
\usepackage{subcaption}
\usepackage{todonotes}
\usepackage[normalem]{ulem}
\usepackage{soul}
\usepackage{commath}
\usepackage{enumitem}

\newtheorem{theorem}{Theorem}[section]
\newtheorem{lemma}[theorem]{Lemma}
\newtheorem{proposition}[theorem]{Proposition}
\newtheorem{corollary}[theorem]{Corollary}

\newtheorem{claim}[theorem]{Claim}
\newtheorem{question}{Question}
\newtheorem{definition}[theorem]{Definition}
\newtheorem{observation}[theorem]{Observation}
\newenvironment{clproof}{\begin{list}{}{%
			\setlength{\leftmargin}{3mm}%
		} \item {\it Proof.} }{\hfill$\lozenge$\end{list}}

\newcommand\extrafootertext[1]{%
    \bgroup
    \renewcommand\thefootnote{\fnsymbol{footnote}}%
    \renewcommand\thempfootnote{\fnsymbol{mpfootnote}}%
    \footnotetext[0]{#1}%
    \egroup
}

\newcommand{\dist}{\mathsf{dist}}
\newcommand{\tw}{\mathsf{tw}}

\newcommand{\rad}{\mathsf{rad}}
\newcommand{\tp}{\mathsf{top}}
\newcommand{\NP}{\textsf{NP}}
\renewcommand\abs[1]{\lvert #1\rvert}
\DeclareTextCompositeCommand{\v}{OT1}{l}{l\nobreak\hspace{-.1em}'}

\title{An efficient algorithm for \textsc{$\mathcal{F}$-subgraph-free Edge Deletion} on graphs having a product structure}
\author[1]{Shinwoo An}
\affil[1]{Department of Computer Science, Bar-Ilan University, Israel.}
\author[2,3]{Seonghyuk Im\thanks{Supported by the National Research Foundation of Korea (NRF) grant funded by the Korea government (MSIT) No. RS-2023-00210430, and supported by the Institute for Basic Science (IBS-R029-C4).}}
\affil[2]{Department of Mathematical Sciences, KAIST, Daejeon, South Korea.}
\affil[3]{Extremal Combinatorics and Probability Group, Institute for Basic Science (IBS), Daejeon, South Korea.}
\author[2,4]{Seokbeom Kim\thanks{Supported by the Institute for Basic Science (IBS-R029-C1).}}
\affil[4]{Discrete Mathematics Group, Institute for Basic Science (IBS), Daejeon, South Korea.}
\author[5]{Myounghwan Lee\thanks{Supported by the National Research Foundation of Korea (NRF) grant funded by the Ministry of Science and ICT (No.~RS-2023-00211670).}}
\affil[5]{Department of Mathematics, Hanyang University, Seoul, South Korea.}

\date{\today}

\begin{document}

\maketitle

\begin{abstract}
Given a family $\mathcal{F}$ of graphs, a graph is \emph{$\mathcal{F}$-subgraph-free} if it has no subgraph isomorphic to a member of $\mathcal{F}$.
We present a fixed-parameter linear-time algorithm that decides whether a planar graph can be made $\mathcal{F}$-subgraph-free by deleting at most $k$ vertices or $k$ edges, where the parameters are $k$, $\lvert \mathcal{F} \rvert$, and the maximum number of vertices in a member of $\mathcal{F}$.
The running time of our algorithm is double-exponential in the parameters, which is faster than the algorithm obtained by applying the first-order model checking result for graphs of bounded twin-width.
    
To obtain this result, we develop a unified framework for designing algorithms for this problem on graphs with a ``product structure.''
Using this framework, we also design algorithms for other graph classes that generalize planar graphs.
Specifically, the problem admits a fixed-parameter linear time algorithm on disk graphs of bounded local radius, and a fixed-parameter almost-linear time algorithm on graphs of bounded genus.
    
Finally, we show that our result gives a tight fixed-parameter algorithm in the following sense:
Even when $\mathcal{F}$ consists of a single graph $F$ and the input is restricted to planar graphs, it is unlikely to drop any parameters $k$ and $\lvert V(F) \rvert$ while preserving fixed-parameter tractability, unless the Exponential-Time Hypothesis fails.
\end{abstract}

\newpage

\section{Introduction}

A graph modification problem is a problem to decide if a graph can be modified by performing a small number of operations, such as adding or deleting vertices and edges.
This problem has gained significant interest (see, e.g.~\cite{Crespelle2023survey}) since it can be used to formulate various algorithmic problems and to mathematically model real-world situations.
For instance, it can be used to model the spread of animal disease~\cite{enright2018deleting}.
A graph represents the transmission network between the farms, where the vertices are the farms and the edges represent the livestock movement between farms. 
The goal is to control the epidemic by prohibiting a small number of these connections and making each connected component have at most $h$ vertices for some integer~$h \geq 1$.

In this paper, we focus on \textsc{$\mathcal{F}$-subgraph-free Edge Deletion} for a family $\mathcal{F}$ of graphs, which is a special type of the graph modification problem that only allows edge deletions.

\medskip
\noindent
\fbox{\parbox{0.97\textwidth}{
	\textsc{$\mathcal{F}$-subgraph-free Edge Deletion}\\
	\textbf{Input:} A graph $G$ and a nonnegative integer $k$\\
	\textbf{Task:} Determine whether there exists $S\subseteq E(G)$ of size at most $k$ such that $G-S$ has no subgraph isomorphic to a member of $\mathcal{F}$
    }}
\medskip

For example, since every connected graph contains a tree that covers all vertices, the above-mentioned example can be described as \textsc{$\mathcal{T}_{h+1}$-subgraph-free Edge Deletion}, where $\mathcal{T}_{h+1}$ denotes the family of trees on $h+1$ vertices.
In fact, there are classic problems, such as \textsc{Edge Bipartization} or \textsc{Feedback Arc Set}, that can be framed as the edge deletion problem.

Several graph modification problems, including previous examples, are \NP-complete~\cite{CTY07,choi1989graph, enright2018deleting, Karp72, Yannakakis1978}.
Even when $\mathcal{F}$ consists of a single graph $F$, the \NP-hardness of \textsc{$\mathcal{F}$-subgraph-free Edge Deletion} has been shown for several graphs $F$.
For example, Anaso and Hirata~\cite{AH82} proved the \NP-completeness when~$F$ is $3$-connected; Alon, Shapira, and Sudakov proved it when $F$ is not bipartite~\cite{ASS09}; and Gishboliner, Levanzov, and Shapira proved it when~$F$ is a forest that is not a star-forest~\cite{GLS24}.
Thus, exploring their complexity with additional constraints has emerged as an important research direction.

One direction is to design \emph{fixed-parameter tractable (FPT)} algorithms, which are algorithms that run in time $f(\textsf{p}) \cdot n^{O(1)}$ for some predetermined parameter $\textsf{p}$.
A natural question in this direction is whether there exists an FPT algorithm when parameterized by the solution size $k$.
For example, Pilipczuk, Pilipczuk, and Wrochna~\cite{pilipczuk2019edge} proved that \textsc{Edge Bipartization} can be solved in time $O(1.977^k) \cdot nm$ on $n$-vertex $m$-edge graphs.
Moreover, if $\mathcal{F}$ is a fixed finite family of graphs, then there is a simple FPT algorithm for \textsc{$\mathcal{F}$-subgraph-free Edge Deletion} when parameterized by $k$, by using the bounded search tree technique~\cite{cygan2015parameterized}. 
However, even when $\mathcal{F} = \{F\}$, we cannot design an FPT algorithm for the same problem parameterized by $k + \lvert V(F) \rvert$ since this problem includes well-known $\textsf{W}[1]$-complete problems such as \textsc{Clique} and \textsc{Independent Set}~\cite{downey1995fixed}, which are unlikely to have FPT algorithms unless the Exponential-Time Hypothesis (ETH) fails~\cite{cygan2015parameterized}.\footnote{The Exponential-Time Hypothesis states that \textsc{$3$-SAT} cannot be solved in time~$2^{o(n)}$.
In particular, ETH implies $\textsf{FPT} \neq \textsf{W}[1]$.}

Another well-studied parameter in parameterized complexity is the \emph{tree-width}, which is a graph parameter that measures how close a graph is to a tree (see~\cref{sec:Prelim} for the formal definition).
Ever since defined by Robertson and Seymour in their Graph Minors series~\cite{RS}, tree-width has become a central topic since it is a powerful tool for designing algorithms.
For instance, the remarkable Courcelle's theorem~\cite{Courcelle90} states that every problem expressible in a particular logic, called the monadic second-order logic, can be solved in linear-time on graphs of bounded tree-width.
However, it again appears to be unsuitable to parameterize \textsc{$\mathcal{F}$-subgraph-free Edge Deletion} by tree-width only, as Gaikwad and Maity~\cite{GM21+} showed that this parameterized problem is $\textsf{W}[2]$-hard when $\mathcal{F}$ is also part of the input.
In contrast, if we parameterize the same problem by $\tw(G) + \lvert \mathcal{F} \rvert + r$, where $r$ denotes the maximum number of vertices in a member of $\mathcal{F}$, then this problem becomes tractable.

\begin{theorem}[{\cite[Theorem~3.1]{enright2018deleting}}] \label{thm:hfed-tw}
    Let $\mathcal{F}$ be a family of graphs where each graph in $\mathcal{F}$ has at most $r$ vertices.
    Then \textsc{$\mathcal{F}$-subgraph-free Edge Deletion} can be solved in time $2^{O(\lvert \mathcal{F} \rvert w^r)} \cdot n$ on $n$-vertex graphs of tree-width at most $w$.
\end{theorem}

Unfortunately, having small tree-width is a rather strong property, and numerous graph classes that arise naturally in algorithmic graph theory have unbounded tree-width.
A standard example is the class of \emph{planar graphs}, which are the graphs that can be drawn in the plane without edge crossings.
Even though planar graphs can have arbitrarily large tree-width~\cite{RS86}, there are \NP-hard problems that can be solved in polynomial-time on planar graphs, such as \textsc{Maximum Cut}~\cite{Hadlock75}.
This has motivated a common research question: which other problems share this property?

Our first result states that \textsc{$\mathcal{F}$-subgraph-free Edge Deletion} is fixed-parameter tractable on planar graphs, and even on graphs of bounded genus (i.e., graphs that can be embedded into a surface of bounded genus) when parameterized by $\lvert \mathcal{F} \rvert + k + r$, where $k$ is the solution size and $r \coloneq \max_{F \in \mathcal{F}} \lvert V(F) \rvert$.

\begin{theorem} \label{cor:F-free-planar}
    Let $\mathcal{F}$ be a family of graphs where each graph in $\mathcal{F}$ has at most $r$ vertices.
    Then \textsc{$\mathcal{F}$-subgraph-free Edge Deletion} can be solved 
    \begin{enumerate}[label=(\roman*)]
        \item in time~$2^{O(\lvert \mathcal{F} \rvert^2 k r^3)^r} \cdot n$ on $n$-vertex planar graphs
        \item in time~$2^{O(\lvert \mathcal{F} \rvert^2 gk r^3)^r} \cdot n \log n$ on $n$-vertex graphs of genus at most $g$.
    \end{enumerate}
\end{theorem}

Note that the tractability of the problem on planar graphs can also be deduced by applying the first-order model checking result on graphs of bounded twin-width~\cite{BKTW20}.
However, the running time of the algorithm obtained in this way has a tower-type dependency on the parameters, compared to a double-exponential dependency in~\cref{cor:F-free-planar}.

Another class we focus on is the class of disk graphs, which is a generalization of planar graphs from a different point of view.
We first give some definitions about disk graphs.
We say a graph $G$ is a \emph{disk graph} if there is a family $\mathcal{D}$ of disks in the Euclidean plane such that $G$ is isomorphic to the intersection graph of $\mathcal{D}$, that is, each disk corresponds to a vertex, and there is an edge between two vertices if and only if two corresponding disks intersect.
In this case, the family $\mathcal{D}$ is called a \emph{geometric representation} of $G$.
Throughout this paper, we assume that the geometric representation $\mathcal D$ is also given as an input, as \textsc{Recognition} is \NP-hard for disk graphs~\cite{breu1998unit}.
Furthermore, we assume that $\mathcal D$ is in \emph{general position} so that
\begin{enumerate}[label=(\arabic*)]
    \item no two disks in $\mathcal{D}$ are identical, and
    \item the boundaries of any three disks in $\mathcal{D}$ do not have a common intersection.
\end{enumerate}
This assumption is valid since one may slightly perturb the disks if necessary.

Disk graphs have an involved structure in general, so it is natural to consider a parameter that measures their complexity.
In this paper, we focus on the \emph{local radius}, which is a parameter measuring the complexity of the disk graph and was recently introduced by Lokshtanov, Panolan, Saurabh, Xue, and Zehavi~\cite{lokshtanov2023framework}.
Roughly speaking, it measures the number of arcs we need to cross to connect two points in a disk. See~\cref{sec:Prelim} for the formal definition.
A class of disk graphs with bounded local radius includes many interesting subclasses. For example, planar graphs are exactly the disk graphs of local radius $1$~\cite{lokshtanov2022framework} and disk graphs of maximum degree at most $\Delta$ have local radius at most $2 \Delta$.
Lokshtanov, Panolan, Saurabh, Xue, and Zehavi~\cite{lokshtanov2023framework} proved that a disk graph with a small local radius admits structural properties analogous to the planar graph, such as bidimensionality theory~\cite{demaine2006bidimensional} or the Grid Minor Theorem~\cite{RS86}.
Motivated by these results, we design an algorithm for \textsc{$\mathcal{F}$-subgraph-free Edge Deletion} on disk graphs of bounded local radius. 
\begin{theorem} \label{cor:F-free-disk}
    Let $\rho \geq 0$ be an integer and let $\mathcal{F}$ be a family of graphs where each graph in $\mathcal{F}$ has at most $r$ vertices.
    Suppose an $n$-vertex disk graph and its geometric representation of local radius $\rho$ are given as an input.
    Then \textsc{$\mathcal{F}$-subgraph-free Edge Deletion} can be solved in time~$2^{O(\lvert \mathcal{F} \rvert^2kr^3\rho^9)^{r}} \cdot n$.
\end{theorem}

\subsection{Product structure of graphs}

To prove \Cref{cor:F-free-planar} and \Cref{cor:F-free-disk}, we provide a unified framework using the \emph{product structure} which simultaneously captures graph classes in both theorems.
In a recent breakthrough paper, Dujmovi\'{c}, Joret, Micek, Morin, Ueckerdt, and Wood~\cite{dujmovic2020planar} initiated the study of graph product structure theory.
Given two graphs $G$ and $H$, the \emph{strong product} of $G$ and $H$, denoted by $G \boxtimes H$, is a graph whose vertex set is $V(G) \times V(H)$ and two vertices $(u, v), (u', v')$ are adjacent in $G \boxtimes H$ if and only if they satisfy one of the following.
\begin{itemize}
    \item $u=u'$ and $vv'\in E(H)$,
    \item $v=v'$ and $uu'\in E(G)$, or
    \item $uu'\in E(G)$ and $vv'\in E(H)$.
\end{itemize}
We say a class $\mathcal{C}$ of graphs \emph{has a product structure} if there is a universal constant $w \geq 1$ such that every member of $\mathcal{C}$ is a subgraph of $H \boxtimes P$ for some graph $H$ of tree-width at most $w$ and a path~$P$.

\begin{theorem}[\cite{dujmovic2020planar}, improved in \cite{ueckerdt2022improved}]
    Every planar graph is a subgraph of $H \boxtimes P$ for some graph~$H$ of tree-width at most $6$ and a path $P$.
\end{theorem}
 This theorem has numerous applications: it not only resolved several long-standing conjectures but also suggested several algorithmic applications on planar graphs. For example:
\begin{itemize}
    \item Dujmovi\'{c}, Joret, Micek, Morin, Ueckerdt, and Wood~\cite{dujmovic2020planar} showed that planar graphs have bounded queue number, which resolves a question by Heath, Leighton, and Rosenberg~\cite{HLR92}. 
    \item Dujmovi\'{c}, Eppstein, Joret, Morin, and Wood~\cite{DEJMW20} showed that planar graphs have bounded nonrepetitive coloring number, which resolves a question by Alon, Grytczuk, Ha{\l{}}uszczak, and Riordan~\cite{AGHR02}.
    \item Dujmovi\'{c}, Esperet, Gavoille, Joret, Micek, and Morin~\cite{DEGJMM21} showed that there is a $(1+o(1))\log n$-bit labeling scheme for $n$-vertex planar graphs. 
    This asymptotically solves a conjecture by Kannan, Naor, and Rudich~\cite{KNR92}.
    \item Esperet, Joret, and Morin~\cite{EJM23} showed that there is a graph on $(1+o(1))n$ vertices and $n^{1+o(1)}$ edges that contains every $n$-vertex planar graphs as subgraphs.
\end{itemize}
Moreover, the certificate for each item can be computed in time $O(n)$~\cite{BMO22}.
Afterward, several classes of graphs are known to have a product structure, including graphs of bounded genus~\cite{distel2022improved, dujmovic2020planar}, apex-minor-free graphs~\cite{dujmovic2020planar}, $k$-planar graphs~\cite{distel2024powers, dujmovic2023graph, hickingbotham2024shallow}, and $d$-map graphs~\cite{bekos2024graph, dujmovic2023graph}.

Building on this context, for a finite family $\mathcal{F}$ of graphs, we design an FPT algorithm for \textsc{$\mathcal{F}$-subgraph-free Edge Deletion} on graphs having a product structure.

We give a high-level approach to our algorithm. 
First, given an input graph $G$, we compute an appropriate product structure $H \boxtimes P$ containing $G$ as a subgraph as well as an embedding of $G$ into $H \boxtimes P$, such that $H$ has tree-width $O(1)$ and $P$ is a path.
Next, we design a framework such that given an input graph as well as a corresponding product structure, solve \textsc{$\mathcal{F}$-subgraph-free Edge Deletion} in linear time. Precisely, we prove the following theorem in~\Cref{sec:f_free_edge_deletion}.
\begin{theorem}\label{thm:prod_to_algo}
    Let $\mathcal{F}$ be a family of graphs where each graph in $\mathcal{F}$ has at most $r$ vertices.
    Then there is an algorithm that, given an integer $k \geq 0$, two graphs $G$ and $H$, and a path $P$ such that
    \begin{itemize}
        \item $G$ has $n$ vertices and $H, P$ has at most $n$ vertices,
        \item $\tw(H) \leq w$, and
        \item $G$ is a subgraph of $H \boxtimes P$,
    \end{itemize}
    solves \textsc{$\mathcal{F}$-subgraph-free Edge Deletion} on $G$ in time $2^{O(\lvert \mathcal{F} \rvert^2kr^3w)^{r}} \cdot n$, assuming that the embedding of $G$ into $H \boxtimes P$ is given.
\end{theorem}

Our main technical contribution comes from~\cref{thm:prod_to_algo}. 
Notably, the problem becomes easy when all graphs of $\mathcal F$ are connected. In this case, we can solve the problem using a standard layering technique (e.g.,~\cite{baker1994approximation}).
However,~\cref{thm:prod_to_algo} works even when graphs in $\mathcal F$ are disconnected. We remark that the previous algorithm of~\cite[Theorem 3.1]{enright2018deleting} also addressed the disconnected case. However, their algorithm was based on dynamic programming on tree decompositions of small width. To the best of our knowledge, no previous result based on a layering scheme addresses edge deletion problems whose target graph is disconnected. 
We believe our technique can be adapted to any layering-like structures and will be of independent interest.

Our insight into handling the disconnected case is as follows (see~\cref{claim:removing_conn_comp}).
Let $C$ be a connected component of a graph $F\in\mathcal F$. 
If there are too many copies of $C$ in $G$ that are ``far'' from each other, then it is impossible to remove all $C$ from the entire graph using a small number of edges. 
Hence, in this case, deleting $F \setminus V(C)$ from $G$ is as hard as deleting $F$ from $G$.

Using this idea, we reduce the problem so that all copies of any connected component are ``close'' to each other. 
Furthermore, we remove ``unnecessary'' parts that are already disjoint from all copies of connected components of graphs in $\mathcal F$.
The resulting graph (roughly) consists of a union of all copies of the connected components of $\mathcal F$.
The diameter of this graph is small since there are only a finite number of connected components in $\mathcal F$.
Subsequently,  we show that the resulting graph has small tree-width.
Finally, we run the algorithm of~\cref{thm:hfed-tw}
as a black-box, which gives an efficient algorithm on \textsc{$\mathcal{F}$-subgraph-free Edge Deletion} for graphs of bounded tree-width. 

\medskip
As our framework works for any graph that has a product structure as well as a corresponding embedding, we can solve \textsc{$\mathcal{F}$-subgraph-free Edge Deletion} for any graph for which we can efficiently find such a product structure and an embedding of the graph into the product structure.
For planar graphs, this is possible due to the work by Bose, Morin, and Odak~\cite{BMO22}.
\begin{theorem}[\cite{BMO22}] \label{thm:planar-product}
    Given a planar graph $G$, there is a linear-time algorithm that finds a strong product $H\boxtimes P$ as well as a corresponding embedding of~$G$ into $H\boxtimes P$, such that $H\boxtimes P$ contains~$G$ as a subgraph, $P$ is a path, and the tree-width of $H$ is at most $6$.
\end{theorem}
Combining~\cref{thm:planar-product} and~\cref{thm:prod_to_algo} verifies~\cref{cor:F-free-planar}-(i).
Similarly, the product structure of graphs of bounded genus can be found in time $O(n \log n)$~\cite{Morin2021AlgoProduct} and this verifies \cref{cor:F-free-planar}-(ii).

Our second contribution establishes a new analog of~\cref{thm:planar-product} for disk graphs with bounded local radius, as no product structure is known for this class. 
\begin{theorem} \label{thm:disk-product-structure}
    Let $\rho \geq 0$ be an integer.
    Then every disk graph of local radius at most $\rho$ is a subgraph of $H \boxtimes P$ where $H$ is a graph of tree-width at most $O(\rho^9)$ and $P$ is a path.
    Furthermore, given a geometric representation of local radius $\rho$, such $H$ and $P$ can be found in time $2^{O(\rho)} \cdot n$.
\end{theorem}
Together with \cref{thm:prod_to_algo}, this proves \cref{cor:F-free-disk}. 
Note that we cannot expect such an extension for general disk graphs since they do not have a product structure~\cite{merker2024intersection}.
On the other hand, \cref{thm:disk-product-structure} covers a known subclass of disk graph that admits product structure. Indeed, $K_t$-free unit disk graphs, known to have product structure~\cite{DHJLW21}, have bounded local radius.
We also note that the product structure gives a rich structure of the graph class, thus \cref{thm:disk-product-structure} implies other structural properties of disk graphs of bounded local radius.
For instance, we have the following result. 

\begin{corollary} \label{cor:local-radius-tww}
    Let $\rho \geq 0$ be an integer. Then every disk graph of local radius at most $\rho$ has twin-width at most $2^{2^{O(\rho^9)}}$. 
\end{corollary}

\cref{cor:local-radius-tww} can be deduced as follows.
We first note the following fact.
Given a graph $G$, we denote by $\Delta(G)$ the \emph{maximum degree} of $G$.

\begin{lemma}[\cite{bonnet2021twin2}] \label{lem:tww_of_products}
    For two graphs $G$ and $H$, we have
    \[
        \textsf{tww}(G \boxtimes H) \leq \max\{\textsf{tww}(G) (\Delta(H)+1) + \Delta(H), \textsf{tww}(H)\} + \Delta(H).
    \]
\end{lemma}

Hlin\v{e}n\'{y} and Jedelsk\'{y}~\cite[Corollary~5.7]{hlinery2025hereditaryproductstructure} proved that if $G$ is a subgraph of the strong product of a graph of tree-width at most $w$ and a path, then $G$ is an \emph{induced} subgraph of the strong product of a graph of tree-width at most $3(w+1)8^{w+1}$ and a path. 
Moreover, Jacob and Pilipczuk~\cite{jacob22twinwidthbound} showed that every graph of tree-width at most $k$ has twin-width at most $3 \cdot 2^{k-1}$.
By combining these results with~\cref{thm:disk-product-structure}, every disk graph of local radius at most $\rho$ is an induced subgraph of $H \boxtimes P$ where $H$ is a graph of twin-width at most $2^{2^{O(\rho^9)}}$ and $P$ is a path.
Since twin-width does not increase when taking induced subgraphs, we deduce that $\textsf{tww}(G) = 2^{2^{O(\rho^9)}}$ by~\cref{lem:tww_of_products}.

\subsection{Tightness of \Cref{cor:F-free-planar}}

Our final result establishes that~\cref{cor:F-free-planar} is tight in the following sense:
Even when $\mathcal{F}$ consists of a single graph, neither the planarity assumption nor any of the parameters~$k$~and~$r$ can be dropped while preserving the fixed-parameter tractability.
The necessity of each of these components, except for the parameter $k$, is demonstrated as follows.

For simplicity, we abbreviate \textsc{$\{F\}$-subgraph-free Edge Deletion} as \textsc{$F$-subgraph-free Edge Deletion}.
With this notion,~\cref{cor:F-free-planar} implies that \textsc{$F$-subgraph-free Edge Deletion} parameterized by $r \coloneq \lvert F \rvert$ and the solution size $k$ is fixed-parameter tractable on planar graphs.

\begin{itemize}
    \item If one can remove $r$ from the parameter, then deciding if a planar graph has a Hamilton cycle can be solve in polynomial-time by taking $k=0$ and $F$ to be the cycle on $n$ vertices, which is impossible unless $\textsf{P} = \NP$~\cite{GJS76}. 
    \item If one can remove the planarity assumption, then \textsc{Independent Set} problem parameterized by the solution size is fixed-parameter tractable by taking $k=0$ and $F$ to be the edgeless graph, which is impossible unless $\textsf{FPT} = \textsf{W}[1]$~\cite[Theorem~13.25]{cygan2015parameterized}.
\end{itemize}

Unlike the previous cases, showing we cannot remove $k$ from the parameter is not direct from the well-known results, and was even not previously known.
Our next result shows that this is indeed impossible.

\begin{proposition} \label{thm:hardness_P4_deletion}
    The problem \textsc{$P_4$-subgraph-free Edge Deletion} restricted to planar graphs is \NP-hard.
\end{proposition}

Our proof for~\cref{thm:hardness_P4_deletion} is not direct, meaning that our reduction does not immediately show the desired result.
We first state what we have actually proved, which would be of independent interest.

\begin{proposition}\label{thm:hardness_K3_factor}
    It is \NP-hard to decide if a given planar graph contains vertex-disjoint triangles that cover all vertices.
\end{proposition}
We remark that without the planarity assumption, Kirkpatrick and Hell~\cite{Kirkpatrick1983factors}  proved the \NP-hardness (see also~\cite{gary1979npcompleteness}). 
In addition, when restricted to planar graphs, Guruswami, Rangan, Chang, Chang, and Wong~\cite{Guruswami1998triangle} proved that computing the maximum number of vertex-disjoint triangles is \NP-hard. 

Our idea to connect~\cref{thm:hardness_P4_deletion} and~\cref{thm:hardness_K3_factor} is inspired by a result on the \emph{Erd\H{o}s-S\'{o}s conjecture} (see~\cite{Erdosongraph1998}), which is one of the oldest conjectures in extremal combinatorics.
This connection was first observed in~\cite{GLS24}.
The conjecture states that if $T$ is a tree on $t$ vertices, then every $n$-vertex $T$-subgraph-free graph contains at most $\frac{t-2}{2}n$ edges. 
Erd\H{o}s and Gallai~\cite{EG59} proved that $P_4$ satisfies the Erd\H{o}s-S\'{o}s conjecture in a stronger sense.
\begin{theorem}[\cite{EG59}]\label{thm:SESC_path}
    For every $t \geq 4$, every $n$-vertex $P_t$-free graph has at most $\frac{t-2}{2}n$ edges. Furthermore, disjoint union of $K_{t-1}$ is the unique $P_t$-free graph with $\frac{t-2}{2}n$ edges.
\end{theorem}
By~\cref{thm:SESC_path}, a planar graph $G$ contains an $n$-edge $P_4$-free subgraph if and only if $G$ contains vertex-disjoint triangles covering all vertices.
Thus, given a planar graph $G$, $G$ contains vertex-disjoint triangles covering all vertices if and only if $G$ can be made $P_4$-subgraph-free by deleting at most $k \coloneq \lvert E(G) \rvert - \lvert V(G) \rvert$ edges.
Therefore,~\cref{thm:hardness_K3_factor} implies~\cref{thm:hardness_P4_deletion}.

\subsection{Related works}

\subparagraph{Induced subgraphs}
One closely related problem is \textsc{Induced $\mathcal{F}$-free Edge Deletion}, which asks to delete at most $k$ edges from a given graph $G$ to make $G$ have no copy of a member of $\mathcal{F}$ as an \emph{induced} subgraph.
This problem has also received significant attention~\cite{Aravind2017Hfreedichotomy, Bliznets2018approximateHfree, Cai1996FPTforhereditary, Cygan2017polykernel, Konstantinidis2021cluster,  Yannakaksis1981cycles}.
We remark that our proof strategy can be extended to \textsc{Induced $\mathcal{F}$-free Edge Deletion} and therefore an analogous theorem of \cref{thm:prod_to_algo} for \textsc{Induced $\mathcal{F}$-free Edge Deletion}, as well as \cref{cor:F-free-planar} and \cref{cor:F-free-disk}, can be proved.

Note that for any family $\mathcal{F}$, let $\mathcal{F}'$ be the class of graphs that contain a member of $\mathcal{F}$ as a spanning subgraph.
Then \textsc{$\mathcal{F}$-subgraph Free Edge Deletion} is equivalent to \textsc{Induced $\mathcal{F}'$-free Edge Deletion}.
Thus, \textsc{Induced $\mathcal{F}$-free Edge Deletion} captures more wider class of problems.
On the other hand, $\mathcal{F}'$ may contains $\Omega(\lvert \mathcal{F}\rvert 2^{\binom{r}{2}}/r!)\geq \lvert \mathcal{F}\rvert 2^{\Omega(r^2)}$ elements. 
Therefore, by considering \textsc{$\mathcal{F}$-subgraph Free Edge Deletion} instead of its induced counterpart, we have an advantage in the running time of the algorithm. 
This allows us to decide, for instance, whether $G$ can be made so that all of its components have size at most $h$ by deleting at most $k$ edges more efficiently.

\subparagraph{Baker's layering technique.}
Baker~\cite{baker1994approximation} introduced a seminal technique that yields polynomial-time approximation schemes (PTAS) for numerous \NP-complete problems on planar graphs. 
The key idea is to partition a planar graph into layers so that any union of a bounded number of consecutive layers has bounded treewidth. 
Since then, the technique has been utilized to design subexponential-time algorithms as well as PTASes for minor-free graphs and bounded-genus graphs~\cite{demaine2008bidimensionality, demaine2006bidimensional,fomin2011bidimensionality}.
It is not hard to see that if a graph has a product structure, we can naturally define a Baker-style layering. 
However, to the best of our knowledge, no prior work has considered \textsc{$\mathcal{F}$-subgraph-free Edge Deletion} within either Baker's layering framework nor graph product structure.

\subparagraph{Twin-width.}
Twin-width is a graph parameter introduced by Bonnet, Kim, Thomass{\'e}, and Watrigant~\cite{BKTW20} defined in terms of contractions on \emph{trigraphs}.
Specifically, a \emph{trigraph} is a triple $G = (V(G), E(G), R(G))$ where $V(G)$ is a finite set and $E(G), R(G)$, are disjoint sets of unordered pairs on $V(G)$. 
We also identify a graph $G$ with a trigraph $(V(G), E(G), \emptyset)$.
Given a trigraph $G$ and two distinct vertices $u, v$, a trigraph $G'$ is obtained from $G$ by \emph{contracting $u$ and $v$} if $V(G') = (V(G) \setminus \{u, v\}) \cup \{w\}$ where $w \notin V(G)$, $G' \setminus w = G \setminus \{u, v\}$, and the neighbors of $w$ in $G'$ are defined as follows.
For each $x \in V(G') \setminus \{w\}$,
\begin{itemize} 
    \item $wx \in E(G)$ if and only if $x$ is adjacent to both $u, v$ in $G$,
    \item $wx \in R(G)$ if and only if $x$ is adjacent to exactly one of $u, v$ in $G$, and
    \item $wx \notin E(G) \cup R(G)$ if and only if $x$ is not adjacent to both $u, v$ in $G$.
\end{itemize}
A \emph{contraction sequence} of a trigraph $G$ is a sequence $G_n, \ldots, G_1$ of trigraphs such that $G_n = G$, $G_1 = K_1$, and $G_{i}$ is obtained from $G_{i+1}$ by contracting a pair of distinct vertices in $G_{i+1}$ for each $i$.
For an integer $d \geq 0$, a contraction sequence $G_n, \ldots, G_1$ is a \emph{$d$-sequence} if the graph $(V(G_i), R(G_i))$ has maximum degree at most $d$ for each $i \in [n]$.
The \emph{twin-width} of~$G$, denoted by $\textsf{tww}(G)$, is the minimum integer $d$ such that a $d$-sequence of $G$ exists.

One of the fundamental results on twin-width is an analog of Courcelle's theorem~\cite{Courcelle90} for another type of logic, called the first-order (FO) logic.
The meta-theorem~\cite{BKTW20} states the existence of an algorithm that, given an FO-sentence $\varphi$ on graphs, an $n$-vertex graph $G$ of twin-width at most $d$, and a $d$-sequence of $G$, decides if $G$ satisfies $\varphi$ in time $f(\lvert \varphi \rvert, d) \cdot n$ for some computable function.
This result gives numerous fixed-parameter linear-time algorithms for various graph problems when restricted to several graph classes, but the function $f$ is non-elementary.

\subparagraph{Product structure on disk graphs.}
Disk graphs exhibit a rich structure, leading to both positive and negative conclusions about the existence of a product structure.
A positive answer comes from \emph{unit disk graphs}, which are the intersection graphs of disks of radius~$1$ in the Euclidean plane.
While the class of all unit disk graphs does not have a product structure due to the existence of large cliques, these are the only obstacles that prohibit the product structure.
For each integer $t \geq 1$, the class of \emph{$K_t$-free} unit disk graphs has a product structure~\cite{DHJLW21}. 

However, if we remove the restriction on the radii of disks, then the structure becomes significantly more complicated.
Observe that the neighborhood of each vertex in a graph with a product structure induces a subgraph of small tree-width.
However, for each integer $n \geq 1$, a graph obtained from an $n$~by~$n$ grid by adding a universal vertex $u$ is a $K_4$-free disk graph, but the neighborhood of $u$ induces a subgraph of large tree-width~\cite{RS86}.
This shows that the class of disk graphs does not have a product structure, even if we assume the $K_4$-freeness.
\cref{thm:disk-product-structure} gives an improved characterization of when disk graphs have a product structure.

\subsection*{Organization}
In~\Cref{sec:Prelim}, we collect some notations and definitions. In~\Cref{sec:f_free_edge_deletion}, we present an algorithmic framework that solves \textsc{$\mathcal{F}$-subgraph-free Edge Deletion} under the given product structure (\Cref{thm:prod_to_algo}). 
In~\Cref{sec:disk_graph}, we show that disk graphs of bounded local radius have a product structure~(\Cref{thm:prod_localradius_disk}), and such a product structure can be found efficiently on those graphs~(\Cref{thm:disk-product-structure}). 
Then, in~\Cref{section:hardness}, we prove that all parameters in our algorithm cannot be dropped to obtain an FPT algorithm (\Cref{thm:hardness_K3_factor}).
Finally, in~\Cref{sec:concluding},
we mention several concluding remarks and open problems.

\subsection*{Model of Computation}
In Section~\ref{sec:disk_graph}, our algorithm uses a lowest common ancestor (LCA) data structure. To implement this efficiently, we use the standard real RAM model~\cite{preparata2012computational}.

\section{Preliminaries} \label{sec:Prelim}
For a positive integer $n$, we denote by $[n]$ the set of all positive integers at most~$n$.

\subparagraph{Basic graph notions.}
We refer to~\cite{Diestel2025Book} for undefined notions.
All graphs in this paper are finite and simple, that is, graphs have neither loops nor multiple edges.
For a graph $G$, we denote by $V(G)$ and $E(G)$ the vertex set and the edge set of $G$, respectively.
A graph $H$ is a \emph{subgraph} of a graph $G$ if $V(H)\subseteq V(G)$ and $E(H)\subseteq E(G)$.
A graph $G$ \emph{contains} a graph $H$ if $G$ has a subgraph isomorphic to $H$; such a subgraph of $G$ is called a \emph{copy} of~$H$.
Given a graph $H$, we say a graph $G$ is \emph{$H$-subgraph-free} if $G$ has no copy of $H$.
Given a family $\mathcal{F}$ of graphs, we also say that a graph $G$ is \emph{$\mathcal{F}$-subgraph-free} if $G$ is $H$-subgraph-free for every~$H \in \mathcal{F}$.
For a set $A$ of vertices in a graph $G$, we denote by $G \setminus A$ the graph obtained from $G$ by removing all vertices in $A$.
We denote by $G[A]$ the induced subgraph of $G$ induced by $A$, that is, $G[A] = G \setminus (V(G) \setminus A)$.
For $v \in V(G)$, we also use the notation $G \setminus v$ to indicate $G \setminus \{v\}$.
Similarly, for a set $S$ of edges in a graph $G$, we denote by $G-S$ the graph obtained from $G$ by removing all edges in $S$.

A graph $G$ is \emph{connected} if there is a $u$-$v$ path for every $u, v \in V(G)$.
A \emph{connected component} of a graph $G$ is a maximal connected subgraph of $G$.
For a graph $G$ and two vertices $u,v\in V(G)$, the 
\emph{distance} between $u$ and $v$ in $G$, denoted by $\dist_G(u,v)$, is defined as the length of a shortest path from $u$ to $v$ in $G$ if it exists; otherwise, we put $\dist_G(u,v)=\infty$.
For a connected graph $G$, the \emph{radius} of $G$, denoted by $\rad(G)$, is defined by $\min_{u\in V(G)}\max_{v\in V(G)}\dist_G(u,v)$.

A \emph{triangle tiling} in a graph $G$ is a family of vertex-disjoint triangles in $G$.
For a vertex $v$ of $G$, we say that a triangle tiling \emph{covers} $v$ if some triangle of it contains $v$.
A \emph{perfect triangle tiling} or \emph{triangle factor} in a graph $G$ is a triangle tiling covering all vertices of $G$.

\subparagraph{Tree-decomposition~\cite{Diestel2025Book}.}
A \emph{tree-decomposition} of a graph $G$ is a pair $(T, \{B_t\}_{t\in V(T)})$ consisting of a tree $T$ and a family of sets $\{B_t\}_{t\in V(T)}$ of vertices in $G$ such that 
\begin{enumerate}[label=(T\arabic*)]
    \item\label{TD:vertex-coverage} $V(G)=\bigcup_{t\in V(T)}B_t$, 
    \item\label{TD:edge-coverage} for every edge $uv \in E(G)$, there exists a node $t \in V(T)$ such that $\{u, v\} \subseteq B_t$, and 
    \item for every vertex $v \in V(G)$, the set $\{t\in V(T):v\in B_t\}$ induces a connected subtree of $T$. \label{TD:connectedness}
\end{enumerate}
The \emph{width} of a tree-decomposition $(T, \{B_t\}_{t \in V(T)})$ is defined as $\max_{t\in V(T)}\abs{B_t}-1$.
The \emph{tree-width} $\tw(G)$ of a graph $G$ is the minimum width among all tree decompositions of $G$.

\subparagraph{Parameterized algorithm.}
We refer the reader to~\cite{cygan2015parameterized} for further details.
A \emph{parameterized problem} is a language $L\subseteq \Sigma^{\ast}\times \mathbb{N}$, where $\Sigma$ is a finite set of alphabets and $\mathbb{N}$ is the set of nonnegative integers.
For an instance $(x,k)\in \Sigma^{\ast}\times \mathbb{N}$, we call $k$ the \emph{parameter}.
A parameterized problem $L\subseteq \Sigma^{\ast}\times \mathbb{N}$ is \emph{fixed-parameter tractable (FPT)} if there exists an algorithm~$\mathcal{A}$, a computable function~$g$, and a universal constant $c \geq 1$ such that given an instance $(x,k)\in \Sigma^{\ast}\times \mathbb{N}$, the algorithm~$\mathcal{A}$ decides whether $(x,k)\in L$ in time $g(k) \cdot (\abs{x}+k)^c$.
We call such an algorithm an \emph{FPT algorithm}.

\subparagraph{Disk graphs.}
The boundaries of the disks in $\mathcal{D}$ partition the Euclidean plane into regions, which are called the \emph{faces} (of $\mathcal{D}$).
The \emph{arrangement graph} of $\mathcal{D}$, denoted by $A_\mathcal{D}$, is defined as follows.
The vertices of $A_\mathcal{D}$ are the faces of $\mathcal{D}$, and two vertices in $A_\mathcal{D}$ are adjacent if and only if they share a common boundary arc which is not a single point.
By abuse of notation, for $D \in \mathcal{D}$, we denote by $A_\mathcal{D}[D]$ the subgraph of $A_\mathcal{D}$ induced by the set of faces in $\mathcal{D}$ contained in $D$ (see~\cref{fig:arrangement}).
Furthermore, we use $v\in V(A_\mathcal D)$ as a vertex of the arrangement graph or a face of $\mathcal D$ interchangeably if it is clear from the context.

\begin{figure}
    \centering
    \includegraphics{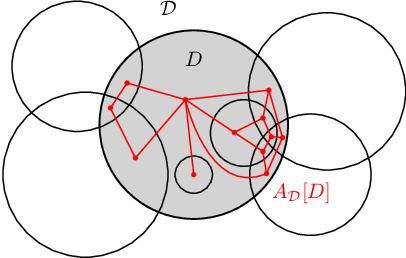}
    \caption{An illustration of an arrangement graph}
    \label{fig:arrangement}
\end{figure}

We now formally define \emph{local radius}.
The \emph{local radius} of $\mathcal{D}$ is the maximum radius of $A_\mathcal{D}[D]$, where the maximum is taken over all $D \in \mathcal{D}$.
Then the \emph{local radius} of a disk graph $G$ is the minimum local radius of the geometric representation of $G$.
A closely related concept of local radius is the notion of \emph{ply}.
The \emph{ply} of a family $\mathcal{D}$ of disks is the largest integer $p$ such that every point in the Euclidean plane lies in at most $p$ disks in $\mathcal{D}$.
Then \emph{ply} of a disk graph $G$ is the minimum ply of the geometric representations of $G$.

\section{Solving \texorpdfstring{$\mathcal{F}$}\sc{-subgraph-free edge deletion} using the product structure}\label{sec:f_free_edge_deletion}
In this section, we give a proof of \Cref{thm:prod_to_algo}. 
We first prove the following simple lemma. It was already observed in~\cite{bodlaender91producttw, hickingbotham25product}, but we give a proof for the completeness.
\begin{lemma}\label{lem:tw_of_product}
    Let $G$ and $H$ be graphs. Then 
    \[\tw(G \boxtimes H) \leq \min\left(\lvert V(G) \rvert \cdot (\tw(H) + 1),\lvert V(H) \rvert \cdot (\tw(G) + 1)\right) - 1.\]
\end{lemma}
\begin{proof}
We first show that $\tw(G \boxtimes H) \leq \lvert V(G) \rvert \cdot (\tw(H) + 1) - 1$.
As the tree-width is decreasing when we take a subgraph, we may assume that $G=K_s$.
Let $V(G)=\set{a_1, a_2,\ldots, a_s}$, and let $(T,\set{B_t}_{t\in V(T)})$ be a tree-decomposition of $H$ of width $\tw(H)$. 
For each $t \in V(T)$, define
\[B_t' = \{(a_i,w) \in V(G \boxtimes H) \mid \text{$w \in B_t$ and $i \in [s]$}\}.\]
We show that 
$(T, \set{B_t'}_{t \in V(T)})$ is a tree-decomposition of $G\boxtimes H$.
Condition~\ref{TD:vertex-coverage} directly follows from the definition of $B_t'$.
Next, we show~\ref{TD:edge-coverage}, suppose two vertices $(a_i, u), (a_j, u')$ are adjacent in $G\boxtimes H$.
Since either $u=u'$ or $uu' \in E(H)$, there is a node $s \in V(T)$ such that $B_s$ contains both $u$ and $u'$.
For such a node $s$, two vertices $(u, a_i)$ and $(u', a_j)$ are contained in $B_s'$. This verifies~\ref{TD:edge-coverage}.
Next, we show~\ref{TD:connectedness}.
Let $(a_i, v)$ be a vertex of $G\boxtimes H$.
Since $(T, \{B_t\})$ is a tree decomposition of $H$, the set $\{t \in V(T) \mid v \in B_t\}$ induces a subtree of $T$.
Subsequently, the set $\{t \in V(T) \mid (a_i, v) \in B_t'\}$ induces a subtree of $T$.
Hence, $(T,\{B'_t\})$ is a tree-decomposition of $G\boxtimes H$.
Finally, since
\[\max_{t\in V(T)}\abs{B_t'}-1=s \cdot(\tw(H)+1)-1,\]
the tree-width of $K_s \boxtimes H$ is at most $s \cdot(\tw(H)+1)-1$.
By changing the role of $G$ and $H$, we have $\tw(G \boxtimes H) \leq \lvert V(H) \rvert \cdot (\tw(G) + 1) - 1$.
This proves the lemma.
\end{proof}

\begin{proof}[Proof of Theorem~\ref{thm:prod_to_algo}]
Let $\ell$ be the number of vertices of $P$.
    We label the vertices of $P$ by $[\ell]$ in natural ordering.
    If $\ell \leq 6r$, then by \Cref{lem:tw_of_product}, $\mathrm{tw}(G) \leq 10rw$. Thus, by \Cref{thm:hfed-tw}, \textsc{$\mathcal{F}$-subgraph-free Edge Deletion} can be solved in time $2^{O( \lvert \mathcal{F} \rvert krw)^r}\cdot n$.
    We now assume that $\ell \geq 6r$.
    For each integer $i \geq 1$, let $V_i = (V(H) \times \{i\}) \cap V(G)$ if $i \leq \ell$; otherwise, let $V_i = \emptyset$.
    We call $V_i$ an $i$-th layer with respect to $H\boxtimes P$. If it is clear from the context, we simply say $V_i$ as an $i$-th layer.
    In addition, for each integer $j \geq 1$, let $I_{j} = [3(j-1)r+1, 3jr]$ and let $V_{I_j} = \bigcup_{i \in I_j} V_i$.

    We first claim that, if there are many distinct indices $j$ such that $G[V_{I_j}]$ contains a copy of a connected component $C$ of some $F \in \mathcal{F}$, then any solution must hit all copies of~$F \setminus C$.
    \begin{claim}\label{claim:removing_conn_comp}
        Let $C$ be a connected component of $F \in \mathcal{F}$ and let $F' = F \setminus V(C)$.
        If there exist at least $k+r$ distinct odd integers $j$ such that $G[V_{I_j}]$ contains a copy of $F'$ as a subgraph, then $(G, k)$ is \textsc{Yes} instance of \textsc{$\mathcal{F}$-subgraph-free Edge Deletion} if and only if $(G, k)$ is \textsc{Yes} instance of \textsc{$(\mathcal{F} \setminus \{F\}) \cup \{F'\}$-subgraph-free Edge Deletion}.
        In particular, if $F$ is connected, then the existence of $k+r$ distinct odd integers $j$ such that $G[V_{I_j}]$ contains a copy of $F$ implies $(G, k)$ is \textsc{No} instance of \textsc{$\mathcal{F}$-subgraph-free Edge Deletion}.
    \end{claim}
    \begin{clproof}
        Let $\widetilde{\mathcal{F}} = (\mathcal{F} \setminus \{F\}) \cup F'$.
        The ``if'' part is clear since every $\widetilde{\mathcal{F}}$-subgraph-free graph is $\mathcal{F}$-subgraph-free.
        To prove the ``only if'' part, assume that there exists a set of edges $S \subseteq E(G)$ of size at most $k$ such that $G-S$ is $\mathcal{F}$-subgraph-free.
        Let $j_1, \ldots, j_{k+r}$ be distinct odd indices such that each $G[V_{I_{j_1}}], \ldots, G[V_{I_{j_{k+r}}}]$ contains a copy of $C$.
        We denote those copies of $C$ by $C_1,C_2,\ldots, C_{k+r}$. 
        Then, as $C_1, \ldots, C_{k+r}$ are pairwise vertex-disjoint, at least $r$ of them do not intersect with $S$.
        Thus, at least $r$ copies among $C_1, \ldots, C_{k+r}$ are contained in $G - S$.
        For the sake of contradiction, suppose that $G-S$ contains a copy $C'$ of $F'$.
        Then each connected component of~$C'$ can intersect at most one of $C_1, \ldots, C_{k+r}$, as the pairwise distance between them is at least~$3r$.
        As $F'$ has at most $r-1$ components, there exists an index $i$ such that $C_i$ is contained in $G-S$ and does not intersect with $C'$.
        However, this shows that $G-S$ contains a copy of $F$, contradicting the assumption that $G-S$ is $\mathcal{F}$-subgraph-free.
        Therefore, we conclude $G-S$ is $\widetilde{\mathcal{F}}$-subgraph-free and this completes the proof of the claim.

        For the second part, we have $k+r\geq k+1$ edge-disjoint copies of $F$ so deleting $k$ edges cannot make $G$ be $F$-subgraph-free.
    \end{clproof}
    \subparagraph{Algorithm.}
    We now describe our algorithm.
    We may assume $\abs{E(G)}> k$ as otherwise we can simply answer \textsc{Yes}.  
    First, we iteratively do the following.
    For each $F \in \mathcal{F}$ and its connected component $C$, we check whether there exist $k+r$ odd indices $j$ such that $G[V_{I_j}]$ contains a copy of $C$.
    If this is the case, we modify $\mathcal F$ by 
    replacing $F$ with $F\setminus V(C)$ if $F\neq C$, and return \textsc{No} if $F=C$.
    We stop the iteration when no such connected component $F'$ exists and perform the next step.

    As a second step, 
    for each odd index $j$, we check if $G[V_{I_j}]$ contains a connected component of some member of $\mathcal F$. 
    If this is the case, we delete the ``middle'' layers $V_{3(j-1)r+r+1}, \ldots, V_{3(j-1)r+2r}$ of $V_{I_j}$ from $G$.
    Letting $G'$ be the graph obtained by removing these layers, we check whether $G'$ can be made $\mathcal{F}$-subgraph-free by deleting at most $k$ edges and return this result.
    This completes our algorithm.

    \subparagraph{Time complexity.}
    Notice that our algorithm operates directly on $G$ (or its induced subgraph), rather than on its supergraph $H \boxtimes P$, whose size can be $\Omega(n^2)$ in the worst case. 
    This gives the benefits in terms of the time complexity. 
    More precisely, since the embedding of $G$ into $H\boxtimes P$ is given, we can extract all information of $G$ and its induced subgraphs $G[V_{I_j}]$'s in time $O(n)$, instead of paying the size of $H\boxtimes P$. Consequently, each step of our algorithm has a linear dependency on $n$.
    
    In the rest of this part, we analyze the time complexity step by step.
    We first analyze the time complexity for the first step.
    Due to~\Cref{lem:tw_of_product}, the tree-width of each $G[V_{I_j}]$ is at most $4rw$.
    Thus, we can perform each iteration in time $2^{O(\lvert\mathcal{F}\rvert rw)^r}n$ by applying the algorithm in~\Cref{thm:hfed-tw} with $k=0$. 
    Initially, $\mathcal F$ consists of at most $\abs{\mathcal F}r$ connected components. Also, each of the first steps decreases the total number of connected components by one. Consequently, the total number of iterations is $O(\abs{\mathcal F}r)$ and therefore the first step takes in time $2^{O(\lvert\mathcal{F}\rvert rw)^r}n$.

    Next, we analyze the complexity for the second step. Since $\tw(G[V_{I_j}])=O(rw)$ by~\cref{lem:tw_of_product}, we can compute $G'$ in time $\abs{P}\cdot 2^{O(\lvert\mathcal{F}\rvert r^2 w)^r}n = 2^{O(\lvert\mathcal{F}\rvert r^2 w)^r}n$.
    As each $F \in \mathcal{F}$ contains at most $r$ connected components, there are at 
    most $\lvert \mathcal{F} \rvert (k+r)r$ odd indices $j$ such that $G[V_{I_j}]$ contains a copy of a connected component of some $F \in \mathcal{F}$.
    In other words, there are no consecutive $\abs{\mathcal F}(k+r)r+1$ odd indices $j$ such that $G[V_{I_j}]$ contains a connected component of a member of $\mathcal F$.
    Also, each~$I_j$ consists of $3r$ consecutive layers of $H\boxtimes P$.
    Thus, 
    each component of $G'$ is a subgraph of $H \boxtimes P_{6 \lvert \mathcal{F} \rvert (k+r)r^2}$.
    Then, by~\Cref{lem:tw_of_product},
    $\tw(G') \leq O(\lvert \mathcal{F} \rvert kr^3 w)$.
    Consequently, due to~\cref{thm:hfed-tw}, \textsc{$\mathcal{F}$-subgraph-free Edge Deletion} can be solved in time $2^{O(\lvert \mathcal{F}\rvert^2r^3kw)^r}\cdot n$ for $G'$.
    In summary, the total time complexity of the algorithm is $2^{O(\lvert \mathcal{F}\rvert ^2r^3kw)^r}\cdot n$.

    \subparagraph{Correctness.}
    Finally, we prove the correctness of the algorithm.
    By~\cref{claim:removing_conn_comp}, the reduction from the first step does not change the solution.
    Thus, it suffices to show that $(G, k)$ is a \textsc{Yes} instance if and only if $(G', k)$ is a \textsc{Yes} instance.
    As $G'$ is a subgraph of $G$, if $G$ can be made $\mathcal{F}$-subgraph-free by deleting at most $k$ edges, then the same is true for $G'$.
    To show the converse, suppose that $(G', k)$ is a \textsc{Yes} instance. 
    Let $S \subseteq E(G')$ be a set of at most $k$ edges such that $G'-S$ is $\mathcal{F}$-subgraph-free.
    We claim that $G-S$ is also $\mathcal{F}$-subgraph-free.
    Assume to the contrary that there exists $F \in \mathcal{F}$ such that $G-S$ contains a copy $C$ of $F$.
    As this copy is not contained in $G'$, $V(C)$ must intersect with a vertex in $V(G) \setminus V(G')$.
    Fix a vertex $v \in V(C) \cap (V(G) \setminus V(G'))$ and denote by $C'$ the component of $C$ that contains $v$.
    Let $i \in [\ell]$ be the index that $v \in V_i$.
    We observe that $V(C') \subseteq \bigcup_{j=i-r}^{i+r} V_j$.
    By the construction, $3(j-1)r+r+1 \leq i \leq 3(j-1)r+2r$ for some odd $j$ and by the above observation, it implies that $C' \subseteq G[V_{I_j}]$. It contradicts that $v \notin V(G')$.
    Therefore, $G-S$ is $\mathcal{F}$-subgraph-free.
    This completes the proof of the correctness.
\end{proof}

\subparagraph{Remark.}
When all the elements of $\mathcal{F}$ are connected, then we can further improve the running time.
To see this, if the number of indices $j$ such that $G[V_{I_j}]$ contains a copy of a graph in $\mathcal F$ is at least $k+1$, then the answer should be \textsc{No} since we have to delete at least one edge from each of those $G[V_{I_j}]$'s.
Using this argument, we can bound the tree-width of $G'$ further, which yields the running time of $2^{O(\lvert \mathcal{F}\rvert krw)^r} \cdot n$.

\section{Product structure for disk graphs of bounded local radius}\label{sec:disk_graph}
In this section, we prove~\cref{thm:disk-product-structure}.
Throughout this section, we denote $G$ as a given disk graph of local radius at most $\rho$.
Furthermore, we denote $\mathcal D$ by the geometric representation of~$G$ having local radius $\rho$.
Note that knowing geometric representation is a natural assumption in an algorithmic sense since \textsc{Recognition} is \NP-hard for disk graphs~\cite{breu1998unit}.
We show that $G$ is a subgraph of $H\boxtimes P$ for a path $P$ and a graph $H$ of small treewidth (\cref{sec:disk-exist},~\cref{thm:prod_localradius_disk}).
Then we present a linear-time algorithm that computes the product structure (\cref{sec:disk-linear-time}).

Recall that the arrangement graph $A_\mathcal D$ is a graph whose vertices are the faces of $\mathcal{D}$, and two faces are adjacent in $A_\mathcal{D}$ if and only if they share a boundary arc (except single point).  
In addition, $A_\mathcal{D}[D]$ is a subgraph of $A_\mathcal{D}$ induced by the faces contained in $D\in \mathcal D$.

\subsection{Existence of the product structure} \label{sec:disk-exist}

\begin{observation} \label{obs:disk_arrange_radius}
    Let $\mathcal{D}, \mathcal{D}'$ be families of disks in the Euclidean plane satisfying $\mathcal{D}' \subseteq \mathcal{D}$.
    Then, for $D \in \mathcal{D}'$, the radius of $A_{\mathcal{D}'}[D]$ is no larger than that of $A_\mathcal{D}[D]$.
\end{observation}

\begin{lemma}[\cite{lokshtanov2023framework}] \label{lem:disk_small_ply}
    Let $G$ be a disk graph of local radius at most $\rho$ for an integer $\rho \geq 0$.
    Then the ply of $G$ is at most $2\rho+1$.
\end{lemma}
\begin{proof}
    We take a geometric representation $\mathcal{D}$ of $G$, whose local radius is at most $\rho$. We show that the ply of $\mathcal D$ is at most $2\rho+1$.
    Assume to the contrary that there is a face of $\mathcal{D}$ which is contained in $2\rho+2$ disks $D_1,D_2,\ldots, D_{2\rho+2}$.
    Let $\mathcal{D}' = \{D_1, \ldots, D_{2\rho+2}\}$.
    By the choice of $\mathcal{D}'$, there is a face~$\sigma$ of $\mathcal{D}'$ that is contained in every disk in $\mathcal{D}'$.
    Take an exposed disk $D_i \in \mathcal{D}'$ with respect to $\mathcal{D}'$.
    That is, the outer boundary of $\bigcup_{D\in \mathcal D'} D$ contains an arc of $D_i$.
    Then there is a face $\sigma'$ of $\mathcal{D}'$ that lies within $D_i$ such that $\sigma'$ is not contained in any disk in $\mathcal{D}'\setminus D_i$. 

    We claim that the distance between $\sigma$ and $\sigma'$ in $A_{\mathcal{D}'}[D_i]$ is at least $2\rho+1$.
    To see this, consider two consecutive faces $\sigma_1,\sigma_2$ on a shortest $\sigma$-$\sigma'$ path in $A_{\mathcal{D}'}[D_i]$.
    By the definition of $A_{\mathcal D'}$, there is a unique disk in $\mathcal{D}'$ that contains one of $\sigma_1, \sigma_2$ but not the other.
    Since $\sigma'$ is contained in exactly one disk in $\mathcal D'$ and $\sigma$ is contained in $\abs{\mathcal D'}=2\rho+2$ disks, at least $2\rho+1$ faces are required to reach $\sigma$ from $\sigma'$. 
    In particular, this implies that the shortest $\sigma$-$\sigma'$ path has length at least $2\rho+1$.
    However, this implies that $\rad(A_{\mathcal{D}'}[D]) \geq \frac{2\rho+1}{2} = r + \frac{1}{2}$, so the local radius of $\mathcal{D}'$ is strictly greater than $\rho$.
    Thus, the local radius of $\mathcal{D}$ is greater than $\rho$ by~\cref{obs:disk_arrange_radius}, which leads to a contradiction.
    This completes the proof of the lemma.
\end{proof}

Our idea is to construct a product structure of $G$ by using that of $A_\mathcal D$. Note that the usage of the arrangement graph has been applied for several algorithms on disk graphs~\cite{an2023faster,berthe2024subexponential,lokshtanov2023framework}. 
To this end, we use the notion of \emph{blow-up} defined as follows.
For $v \in V(G)$, the \emph{$(v, t)$-blow-up} is an operation that replaces $v$ by a clique of size $t$ and adds edges between each vertex in the clique and all neighbors of $v$.
By using this operation, we define the graph $A_\mathcal{D}'$ from $A_\mathcal D$ by applying a $(v, p_v)$-blow-up for each $v \in V(A_\mathcal{D})$, where $p_v$ denotes the number of disks in $\mathcal{D}$ containing $v$.
See~\cref{fig:A'D} for an illustration.
It is not hard to observe that blow-up operation can be captured by strong products.
\begin{observation} \label{obs:blow-up-product}
    Let $G$ be a subgraph of $H$ and let $t \geq 1$ be an integer.
    Let $G'$ be the graph obtained from $G$ by applying a $(v,t)$-blow-up for each $v\in V(G)$.
    Then $H\boxtimes K_t$ contains $G'$ as a subgraph.
\end{observation}
\begin{proof}
    For each $v \in V(G)$, let $v^{1}, \ldots, v^{t}$ be the vertices in $G'$ that are created by applying the $(v, t)$-blow-up for $v$ in $G$ and let $C_v = \{v^{1}, \ldots, v^{t}\}$.
    In particular, $C_v$ is a clique of size $t$ in $G'$ for each $v$ and $V(G') = \bigcup_{v \in V(G)} C_v$.

    To show that $H \boxtimes K_t$ contains $G'$ as a subgraph, let $V(K_t) = \{a_1, \ldots, a_t\}$.
    Define a function $\phi : V(G') \to V(H \boxtimes K_t)$ by $\phi(v^{i}) = (v, a_i)$.
    We claim that $\phi$ gives an isomorphism from $G'$ to a subgraph of $H \boxtimes K_t$.
    Take an edge $u^{i} v^{j} \in E(G')$ for $u, v \in V(G)$ and $i, j \in [t]$.
    Then either $uv \in E(G)$, or $u=v$ and $i \neq j$ holds.
    Since $a_i a_j \in E(K_t)$ whenever $i \neq j$, this implies that $(u, a_i)$ and $(v, a_j)$ are adjacent in $H \boxtimes K_t$.
    Therefore, we conclude that $H \boxtimes K_t$ contains $G'$.
\end{proof}

\begin{figure}
    \centering
    \includegraphics[scale=0.4]{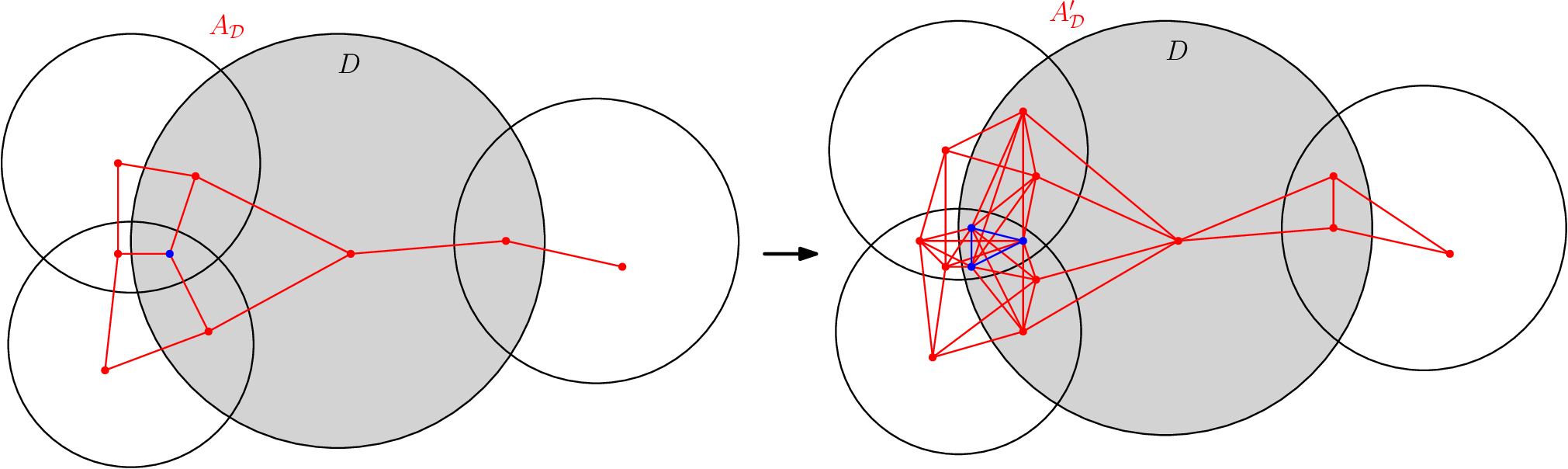}
    \caption{An illustration of $A'_\mathcal{D}$.}
    \label{fig:A'D}
\end{figure}

Next, we give an overview of our proof of~\cref{thm:prod_localradius_disk}.
Our proof consists of two parts.
The first part is to compute a product structure of $A_\mathcal{D}'$ using that of $A_\mathcal D$. 
Since this part is relatively simple, we defer its detail to the end of this section.
The second part is to construct a product structure of $G$ using that of $A_\mathcal{D}'$.
To do this, we first introduce the notion of \emph{depth-$d$ minor model}, and verify that $G$ can be viewed as a depth-$\rho$ minor of $A'_\mathcal D$.
\begin{definition}[Minor model~\cite{sparsity2012}]
    A \emph{minor model} of a graph $H$ in a graph $G$ is a function $\phi : V(H) \to 2^{V(G)}$ that satisfies the following properties:
\begin{enumerate}[label=(\roman*)] \setlength\itemsep{-0.1em}
    \item the collection $\{\varphi(v) : v \in V(H)\}$ is pairwise disjoint, \label{branch-set-disjoint}
    \item for each $v \in V(H)$, the set $\phi(v)$ induces a connected subgraph of $G$, and \label{branch-set-connectivity}
    \item for each $uv \in E(H)$, there is $xy \in E(G)$ such that $x \in \phi(u)$ and $y \in \phi(v)$. \label{adjacency-preserving}
\end{enumerate}
For $v \in V(H)$, the set $\phi(v)$ is called the \emph{branch set} of $v$.
A graph $H$ is a \emph{minor} of a graph $G$ if there is a minor model of $H$ in $G$.
For an integer $d \geq 0$, a minor model $\phi$ of $H$ in $G$ is called a \emph{depth-$d$ minor model} if it additionally satisfies that 
\begin{enumerate}[label=(iv)]
    \item $\rad(G[\phi(v)]) \leq d$ for each $v \in V(H)$. \label{small-radius}
\end{enumerate}
We say a graph $H$ is a \emph{depth-$d$ minor} of $G$ if there is a depth-$d$ minor model of $H$ in $G$.
\end{definition}

\begin{theorem}[{\cite[Theorem~11]{hickingbotham2024shallow}}]\label{Thm:Shallow}
Let $d,w\ge 0$ and $\ell\ge 1$ be integers.
If a graph $G$ is a depth-$d$ minor of $H\boxtimes P\boxtimes K_{\ell}$ for some graph $H$ of tree-width at most $w$ and some path $P$, then $G$ is a subgraph of $H'\boxtimes P\boxtimes K_{\ell(2d+1)^2}$ for some graph $H'$ of tree-width at most ${2d+w+1 \choose w}-1$ and some path~$P$.
\end{theorem}
We show that there is a depth-$\rho$ minor model of $G$ in $A_\mathcal{D}'$ where $\rho$ is the local radius of $G$.

\begin{lemma} \label{lem:blowup_minor}
    Let $G$ be a disk graph and let $\mathcal{D}$ be a geometric representation of $G$ of local radius $\rho$ for an integer $\rho \geq 0$.
    Then $A_\mathcal{D}'$ contains $G$ as a depth-$\rho$ minor.
\end{lemma}
\begin{proof}
For each $v \in V(G)$, let $D_v$ be the disk of $\mathcal{D}$ corresponding to $v$ and let $\Sigma_v$ be the set of faces of $\mathcal{D}$ that are contained in $D_v$.
    Recall that for each face $\sigma \in \Sigma_v$, a clique of size $p_\sigma$ is associated to $\sigma$ in $A_\mathcal{D}'$.
    Now, define the branch set $\phi(v)$ by taking a vertex from each clique that is associated to a face in $\Sigma_v$.
    Observe that since we substituted a clique of size $p_\sigma$ for each $\sigma$, we can take each set $\phi(v)$ to be vertex-disjoint so that $\phi$ satisfies~\ref{branch-set-disjoint}.
    Let $\psi$ be an embedding of $G$ into $H_1\boxtimes H_2\boxtimes\ldots \boxtimes H_k$. Then for each $i$, we denote the \emph{projection} of $\psi$ to $H_i$ by $\psi_{H_i}$. That is, $\psi_{H_i}(v)=v_i$ where $\psi(v)=(v_1,v_2,\ldots,v_k)$.

    Now, to prove that the function $\phi$ indeed gives a depth-$\rho$ minor model of $G$ in $A_\mathcal{D}'$, we show that $\phi$ satisfies~(ii),~(iii), and~(iv) with $d=\rho$.
    Since each vertex in $\phi(v)$ corresponds to a face of $\mathcal{D}$ contained in $D_v$, the following property holds:
    \begin{enumerate}[label=($\star$)]
        \item The subgraph of $A_\mathcal{D}'$ induced on $\phi(v)$ is isomorphic to $A_\mathcal{D}[D]$. \label{same-induced-subgraph}
    \end{enumerate}
    In particular, this shows that the function $\phi$ satisfies~(ii).
    To show that $\phi$ satisfies~(iii), take an edge $uv \in E(G)$.
    Then there is a face $\tau$ of $\mathcal{D}$ contained in both $D_u$ and $D_v$.
    By the choice of $\phi$, there are vertices $u', v'$ in the clique in $A_\mathcal{D}'$ corresponding to $\tau$ such that $u' \in \phi(u)$ and $v \in \phi(v)$.
    Since the blow-up operation preserves the adjacency between the vertices, $u'v' \in E(A_\mathcal{D}')$ so $\phi$ satisfies~(iii).
    Finally, since the local radius of $\mathcal{D}$ is $\rho$, the radius of $A_{\mathcal{D}}'[\phi(v)]$ is at most $\rho$ by~($\star$).
    Therefore, we conclude that $\phi$ is a depth-$\rho$ minor model of $G$ in $A_\mathcal{D}'$.
\end{proof}
By combining all previous results, we finally derive the product structure for disk graphs having a small local radius. 
\begin{theorem} \label{thm:prod_localradius_disk}
    Let $\rho \geq 0$ be an integer.
    Then every disk graph $G$ of local radius at most $\rho$ is a subgraph of $H \boxtimes P$, where $\tw(H) = O(\rho^9)$ and $P$ is a path.
\end{theorem}
\begin{proof}
Due to~\cref{thm:planar-product}, the arrangement graph $A_\mathcal{D}$ of $\mathcal{D}$ is a subgraph of $J \boxtimes P$ for some graph~$J$ of tree-width at most $6$ and a path $P$.
Moreover, by~\cref{lem:disk_small_ply}, each vertex in $A_\mathcal{D}$ is contained in at most $2\rho+1$ distinct disks in $\mathcal{D}$.
Let $\widetilde{A_{\mathcal{D}}}$ the graph obtained from $A_\mathcal{D}$ by applying a $(v, 2\rho+1)$-blow-up for each $v \in V(A_\mathcal{D})$.
Then $A_\mathcal{D}'$ is a subgraph of $\widetilde{A_\mathcal D}$ and therefore, $A_\mathcal{D}' \subseteq J \boxtimes P \boxtimes K_{2\rho+1}$ by~\cref{obs:blow-up-product}.
Thus, by~\cref{lem:blowup_minor} and~\cref{Thm:Shallow}, $G$ is a subgraph of $J' \boxtimes P \boxtimes K_{(2\rho+1)^3}$ for some graph $J'$ of tree-width at most ${2\rho+7\choose 6}-1$.
Let $H \coloneqq J' \boxtimes K_{(2\rho+1)^3}$. 
Then $G$ is a subgraph of $H\boxtimes P$, and $\tw(H) \leq \left({2\rho+7\choose 6}-1\right)\cdot(2\rho+1)^3=O(\rho^9)$ by~\cref{lem:tw_of_product}.
\end{proof}

\subsection{Computing the product structure in linear time}
\label{sec:disk-linear-time}
In this section, we present a linear-time algorithm that computes the product structure specified in~\cref{thm:prod_localradius_disk} as well as an appropriate embedding of the input graph into the product structure.
Here, \emph{an embedding of $G$ into $H$} is defined as an injective homomorphism\footnote{A mapping that preserves the adjacency} from $G$ to $H$. The following is the main result of this section.

\begin{lemma}\label{lem:algo_product_disk}
    Let $G$ be a disk graph such that a geometric representation $\mathcal D$ of $G$ with local radius~$\rho$ is given. 
    Then one can compute two graphs $H$ and $P$ such that $\tw(H) = O(\rho^9)$, $P$ is a path, and $G$ is a subgraph of $H \boxtimes P$, in time $2^{O(\rho)}\cdot n$.
\end{lemma}
\begin{proof}
We sequentially compute the following graphs as well as appropriate strong products corresponding to them, which appear in the proof of~\cref{thm:prod_localradius_disk}, in linear time.
\begin{itemize}  \setlength\itemsep{-0.1em}
    \item (i) the arrangement graph $A_\mathcal D$ (\cref{obs:algo_arrangement_graph});
    \item (ii) $A_\mathcal{D}'$ obtained from $A_\mathcal D$ by blow-up operations (\cref{claim:algo_disk_blow}); and
    \item (iii) A depth-$\rho$ minor model of $G$ to $A_\mathcal{D}'$ (\cref{claim:algo_disk_depth_r}).
\end{itemize} 

\begin{claim} \label{obs:algo_arrangement_graph}
    We can compute $A_\mathcal D$ in time $O(\rho n)$. Furthermore, we can compute a product structure $H\boxtimes P$ containing $G$ with the embedding of $\mathcal A_D$ into $H\boxtimes P$ in time $O(\rho n)$.
    Here, $H$ is a graph of $\tw(H)\leq 6$, and $P$ is a path.
\end{claim} 
\begin{clproof}
        First, we show the first part.
    The ply of $\mathcal D$ is at most $2\rho+1$ by~\cref{lem:disk_small_ply}. Subsequently, the complexity of the arrangement graph $A_\mathcal D$, which is $\abs{V(A_\mathcal D)}+\abs{E(A_\mathcal D)}$, is $O(\rho n)$.
    To see this, let $V_\mathcal D$, $E_\mathcal D$, and $F_\mathcal D$ be the set of vertices, arcs, and faces of $\mathcal D$, respectively. 
    For a disk $D\in \mathcal D$, at most $O(\rho)$ disks of $\mathcal D$ having radii greater than the radius of $D$ can intersect $D$. This means that $\abs{E(G)}=O(\rho n)$ because one can charge each $uv\in E(G)$ into $u$ in the case that the radius of $u$ is smaller than that of $v$.
    Since each pair of disk boundaries intersects at most twice, $\abs{V_\mathcal D}=O(\rho n)$.
    Since each vertex of $V_\mathcal D$ is an intersection point of four arcs of $E_\mathcal D$,
    $2\abs{E_\mathcal D}=4\abs{V_\mathcal D}$. Subsequently, $\abs{E_\mathcal D}=O(\rho n)$.
    Due to the Euler's formula that $\abs{V_\mathcal D}-\abs{E_\mathcal D}+\abs{F_\mathcal D}=2$, $\abs{F_\mathcal D}=O(\rho n)$.
    Note that the vertex set of $A_\mathcal D$ is exactly $F_\mathcal D$, and each edge of $A_\mathcal D$ corresponds to an arc of $E_\mathcal D$.
    Therefore, the complexity of $A_\mathcal D$ is $O(\rho n)$.
    Furthermore, it is not hard to see that obtaining the planar graph $A_\mathcal D$ together with its embedding in the plane can be done in time linear in the complexity of $A_\mathcal D$, which is $O(\rho n)$.

    For the second part, since $A_\mathcal D$ is a planar graph, we can compute product structure and corresponding embedding in time $O(\abs{V(A_\mathcal D)})=O(\rho n)$ due to~\cref{thm:planar-product}. 

\end{clproof}

Note that each edge of $G$ corresponds to an arc of $A_\mathcal D$. 
Then,~\cref{obs:algo_arrangement_graph} implicitly shows that the number of edges in $G$ is $O(\rho n)$. 
This sparsity provides the key insight that allows us to design a linear-time ($O_\rho(n)$) algorithm.
\begin{claim} \label{claim:algo_disk_blow}
    We can compute $A'_\mathcal D$ in $O(\rho^3 n)$ time.
    Furthermore, let $\psi$ be a given embedding of $A_\mathcal D$ into $H\boxtimes P$ specified in~\cref{obs:algo_arrangement_graph}.
    Then we can compute an embedding $\psi'$ of $A_\mathcal{D}'$ into $H\boxtimes P\boxtimes K_{2\rho+1}$ in time $O(\rho^2 n)$.
\end{claim}
\begin{clproof}
    We first show the first part that computes $A_\mathcal D'$.
    First, we replace each $v\in V(A_\mathcal D)$ with a clique $K_v$ of size $p_v$ in $O(\rho^2 n)$ times in total, as $p_v=O(\rho)$ by~\cref{lem:disk_small_ply}.
    Then, for each edge $uv\in E(A_\mathcal D)$, we add all possible edges between $V(K_u)$ and $V(K_v)$. This takes at most $O(\rho^2)$ time per edge and therefore computing $A_\mathcal{D}'$ takes $O(\rho^3 n)$ time in total. This verifies the first part.

    For the second part, from the proof of~\cref{thm:prod_localradius_disk}, $A_\mathcal{D}'$ is a subgraph of $H\boxtimes P\boxtimes K_{2\rho+1}$.
    Note that each vertex $v$ of $A_\mathcal D$ corresponds to the clique $K_v$ of size $p_v\leq 2\rho+1$ in $A_\mathcal D$.
    We define a one-to-one mapping $\psi'$ as follows.
    For a clique $K_v$ constructed from the blow-up operation on $v\in V(A_\mathcal D)$, 
    $\psi'$ maps the vertices of $K_v$ into $(\psi(v), i)$ for $i\in [2\rho+1]$. Since the size of $K_v$ is at most $2\rho+1$, we can appropriately construct a one-to-one mapping $\psi'$.
    Obviously, computing $\psi'$ takes $\abs{V(A_\mathcal{D}')}=O(\rho^2 n)$ time.
        
    We show that $\psi'$ is an embedding of $A_\mathcal{D}'$ into $H\boxtimes P\boxtimes K_{2\rho+1}$. First, all but the last coordinate of $\psi'$ preserves the adjacency since $\psi$ is an embedding of $A_\mathcal D$ into $H\boxtimes P$.
    Furthermore, since the last coordinate corresponds to $K_{2\rho+1}$, the last coordinates of $\psi(u)$ and $\psi(v)$ form an edge of $K_{2\rho+1}$ for every vertex pair $u$ and $v$. 
\end{clproof}
    
\begin{claim} \label{claim:algo_disk_depth_r}
    Let $\psi'$ be a given embedding of $A_\mathcal{D}'$ into $H\boxtimes P\boxtimes K_{2\rho+1}$ specified in~\cref{claim:algo_disk_blow}. Then we can compute a graph $J$ such that $G\subseteq J\boxtimes P$, $P$ is a path, and $\tw(J)=O(\rho^9)$, in time $2^{O(\rho)} \cdot n$.
\end{claim}
\begin{clproof}
    We show that we can construct the product structure $H'\boxtimes K_{(2\rho+1)^3}\boxtimes P$ of~\cref{Thm:Shallow} in time $2^{O(\rho)} \cdot n$.
    Recall that the proof of~\cref{Thm:Shallow} 
    (\cite[Theorem~7, Theorem~11]{hickingbotham2024shallow}) consists of two steps. First, they compute a so-called normalized tree decomposition $(T,\beta)$~\cite{dujmovic2023graph} of $H$ such that
    \begin{itemize} \setlength\itemsep{-0.1em}
        \item the tree-width of $(T,\beta)$ is $O(\tw(H))$,
        \item $T$ is a rooted tree with $V(T)=V(H)$, and
        \item For each $x\in V(H)$, the subtree $T[x]\coloneqq \{y\in V(T)\mid x\in \beta(y)\}$ is rooted at $x$.
    \end{itemize}
    Second, given a normalized tree decomposition and a minor model $\phi$ of $G$ to $A_\mathcal{D}'$, they compute $H'$ as follows.
    Let $\bar H\coloneqq H\boxtimes K_{2\rho+1}$ and let $\bar \psi$ be an embedding of $A_\mathcal{D}'$ into $\bar H\boxtimes P$ such that $\bar \psi \equiv \psi'$.
    Recall that $\bar \psi_H$ is a projection of $\bar \psi$ to $H$.
    For each branch set $\phi(v)$, consider the set 
    $\bar \psi_H(\phi(v)) \coloneqq \{ \bar \psi_H(x) \mid x \in \phi(v) \}$. 
    Let $\tp(v)$ be the highest node among all nodes of $T$ such that $\bar \psi_H(\phi(v))\cap \beta(x) \ne \emptyset$.
    Define the new branch set $\phi'(u)$ as the union of all branch sets $\phi(v)$ such that $\tp(v)=u$.
    Then they set $H'$ whose vertex set is the set of $\phi'(u)$'s, and edge set is $\{\phi'(u)\phi'(v)\mid \text{there is an edge of } A_\mathcal{D}' \text{ between } \phi'(u), \phi'(v)\}$.

    We show that these steps can be done in linear time. 
    For the first step, we compute a (standard) tree decomposition $(T',\beta')$ of $H$ having tree-width $O(\tw(H))$ in $2^{O(\tw(H))}n$ time using tree-width approximation algorithms 
    (e.g., \cite{bodlaender1993linear,korhonen2022twapprox}).
    We remark that the transformation described in~\cite{dujmovic2023graph} can be done in linear time with respect to the total size of all bags. Indeed, each step—such as rooting subtrees, duplicating nodes, and contracting auxiliary nodes—can be carried out by standard root-to-leaf tree traversals. Therefore, the overall procedure runs in time linear in the total bag size of $(T',\beta')$, which is $O(\tw(H)n)$.
    For the second step, the main bottleneck is to compute $\tp(v)$ for all branch sets $\phi(v)$. This can be computed by applying the LCA (lowest common ancestor) query $\abs{\phi(v)}$ times. By applying the well-known optimal data structures for this subproblem (e.g.,~\cite{bender2000lca}), we can compute all $\tp(v)$ in $O(n)+O\left(\sum_v \abs{\phi(v)}\right)=O(\rho n)$ time. After that, we can compute $H'$ in $\abs{E(A_\mathcal{D}')}=O(\rho^3 n)$ time by traversing each edge of $A_\mathcal{D}'$ and insert edges into $H'$ appropriately.    
    By setting $J \coloneqq H'\boxtimes K_{(2\rho+1)^3}$ and the fact that $\tw(H')=O(\rho^6)$, we can compute $J$ with the desired properties in time $2^{O(\rho)} \cdot n$.   
\end{clproof}
\end{proof}
Finally, we conclude the proof of~\cref{thm:disk-product-structure} by plugging~\cref{thm:prod_localradius_disk} and~\cref{lem:algo_product_disk}.
    
\section{Hardness for \texorpdfstring{$\mathcal{F}$}\textsc{-subgraph-free edge deletion}} \label{section:hardness}

In this section, we prove~\cref{thm:hardness_K3_factor}.
We prove this by a reduction from \textsc{Planar $1$-in-$3$-SAT}, which is known to be \NP-complete~\cite[Lemma~2.1]{dyer1986planar}.

\medskip
\noindent
\fbox{\parbox{0.97\textwidth}{
	\textsc{Planar $1$-in-$3$-SAT}\\
	\textbf{Input :} A 3-CNF formula $\varphi$ whose incidence graph $G_\varphi$ is a planar graph.\\
	\textbf{Task :} Determine whether there exists a satisfying assignment for $\varphi$ such that exactly one literal in each clause is assigned by $\textsc{True}$.
    }}
\medskip

More specifically, given a 3-CNF formula $\varphi$ as well as the planar incidence graph $G_\varphi$, we construct a planar graph $G_\varphi'$ so that $G_\varphi'$ has a triangle factor if and only if there is a satisfying assignment for $\varphi$ such that exactly one literal in each clause of $\varphi$ is assigned by \textsc{True}.
We introduce a high-level overview of our reduction. 
We first replace each variable vertex $x_i$ of $G_\varphi$ with a ``sun-like structure'' $X_i'$, which is obtained by starting from an even cycle and attaching a triangle for each edge. 
Note that to cover all the vertices of the even cycle, we must select all the ``odd-th triangles'' or ``even-th triangles'' which will correspond to \textsc{True} or \textsc{False}, respectively.
For each clause $C_i$, we replace it with a gadget $D_i$ with the property that it has three disjoint vertex sets $U_1, U_2, U_3$ on its outer face such that $D_i-\bigcup_{i \in I} U_i$ has a perfect triangle tiling if and only if $\abs{I}=2$.
Finally, we shall connect those gadgets so that if $x_i \in C_j$, then connect the odd-th triangle and $D_j$ and if $\lnot x_i \in C_j$, then connect the even-th triangle and $D_j$.

Here are two technical issues.
First, note that each of $U_1$, $U_2$, and $U_3$ consists of three vertices, however, we cannot connect three consecutive triangles of $X_i'$ to $D_j$ simultaneously. 
Second, if we connect only the odd-th triangles or only the even-th triangles to clause gadgets, then some vertices of $X_i'$ remain uncovered. 
We resolve these issues by constructing an additional gadget, called \emph{splitter}, which makes four copies of a vertex in $X_i'$. 
More specifically, we first attach a splitter and we connect three of the copied vertices to $D_j$ and connect  the remaining copied vertex to the remaining vertices of $X_i'$.
See \cref{fig:ReductionIllust} for an example.

\begin{figure}
    \centering
    \includegraphics[scale=1]{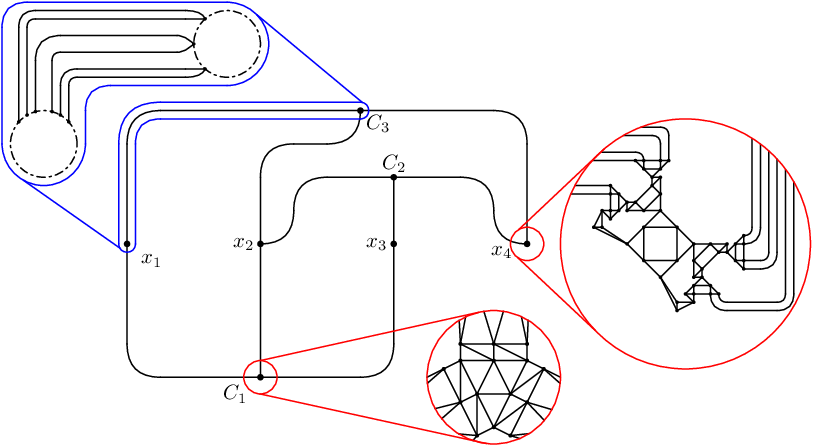}
    \caption{An example of our reduction. Here, the given formula is $\varphi\equiv C_1\land C_2\land C_3$, where $C_1\equiv x_1\lor x_2\lor x_3$, $C_2\equiv x_2\lor \lnot x_3\lor x_4$, and $C_3\equiv x_1\lor x_2\lor \lnot x_4$, and $\mathcal{X}=\set{x_1,x_2,x_3,x_4}$.}
    \label{fig:ReductionIllust}
\end{figure}

\subparagraph{Splitter gadget.}
We first define $L(g, g_1, g_2)$, an auxiliary graph $L$ together with specified vertices $g$, $g_1$, $g_2$ as follows:
\begin{itemize}
    \item $V(L)=\set{g,g_1,g_2, h_1^1,h_2^1, h^2_1, h^2_2, h^2_3, h^2_4}$; and
    \item $E(L)=\set{gh^1_1, gh^1_2, h^1_1h^1_2, h^1_1h^2_1, h^1_1h^2_2, h^1_2h^2_3, h^1_2h^2_4, h^2_1h^2_2, h^2_2h^2_3, h^2_3h^2_4, h^2_1g_1, h^2_2g_1, h^2_3g_2, h^2_4g_2}$.
\end{itemize}
Next, we define the \emph{splitter} $H^{\mathrm{sp}}(g, h, h_1, h_2, h_3)$, which is a graph $H^{\mathrm{sp}}$ together with specified vertices $g$, $h$, $h_1, h_2$ and $h_3$.
Intuitively, these specified vertices are the only vertices of the splitter that will be connected to other gadgets.
We construct $H^{\mathrm{sp}}$ as follows.
Starting from one copy of $L(g,g_1,g_2)$ and two vertex-disjoint copies $L(g^1, h_1, h_2)$ and $L(g^2, h_3, g_2^2)$ of $L(g, g_1, g_2)$, we
\begin{itemize}
    \item identify $g_i$ with $g^i$ for each $i\in \{1, 2\}$;
    \item add three vertices $w_1$, $w_2$, and $h$; and
    \item add the edges $g^2_2w_1$, $g^2_2w_2$, $w_1w_2$, $w_1h$, and $w_2 h$.
\end{itemize}
See~\cref{fig:Splitter} for an illustration. 
\begin{figure}[h]
    \centering
    \includegraphics[scale=0.7]{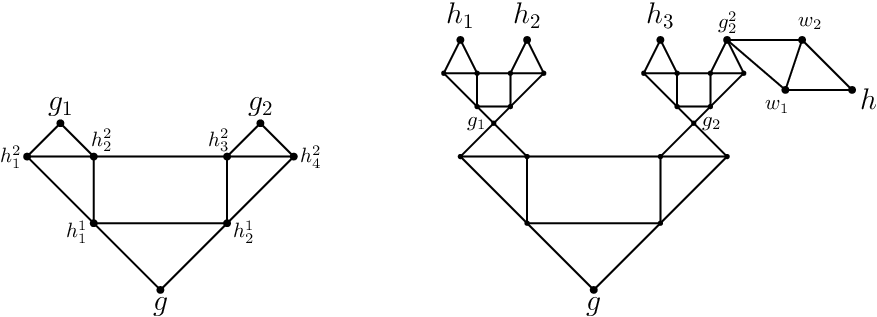}
    \caption{An illustration of $L(g,g_1,g_2)$ (left) and a splitter $H^{\mathrm{sp}}(g, h, h_1, h_2, h_3)$ (right).}
    \label{fig:Splitter}
\end{figure}

Let $\varphi\equiv C_1\land C_2\land\cdots \land C_m$ be a CNF formula such that each clause $C_i$ has exactly three literals.
Let $\mathcal{X}=\set{x_1, x_2, \ldots, x_n}$ be the set of all variables in $\varphi$ and $\mathcal{C} = \{C_1, \ldots, C_m\}$ be the set of all clauses.
Recall that the incidence graph $G_\varphi$ of $\varphi$ is planar. We fix a planar embedding of $G_\varphi$.
We begin by describing the variable and clause gadgets, and then explain how the entire graph is transformed into $G_\varphi'$.

\subparagraph{Variable gadget.}
For a variable $x_i\in\mathcal{X}$, let 
$\ell$ be its degree in $G_\varphi$ and let $C_{1}^i, C_{2}^i, \ldots, C_{\ell}^i \in \mathcal{C}$ be the clauses adjacent to $x_i$ in $G_\varphi$.
Here, the indices are assigned according to the cyclic order in the planar drawing of $G_\varphi$.
We construct an auxiliary graph $X_i'$ as follows. Starting from the cycle on the vertex set  $\{ v_1^i, v_2^i,\ldots, v_{4\ell}^i \}$, we add edges $v_j^i v_{j+2}^i$ for all odd $j \in [4\ell]$ where $v_{4\ell+1}^i=v_1^i$.
\begin{figure}[h]
    \centering
    \includegraphics[scale=0.6]{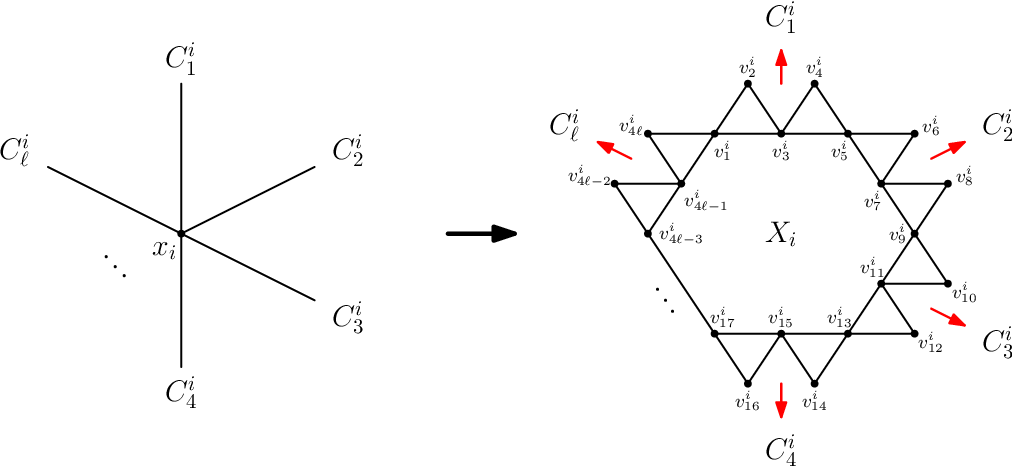}
    \caption{An illustration of a variable gadget.}
    \label{fig:GadgetVariable}
\end{figure}

Next, we define the variable gadget $X_i$ based on $X_i'$. For each integer $j\in [\ell]$, we consider a copy $H^{i, j}(g^{i, j}, h^{i, j}, h_1^{i, j}, h_2^{i, j}, h_3^{i, j})$ of $H^{\mathrm{sp}}$ with specified vertices $g^{i, j}, h^{i, j}, h_1^{i, j}, h_2^{i, j}$, and $h_3^{i, j}$. 
In the rest of this section, we write $H^{i,j}\coloneqq H^{i, j}(g^{i, j}, h^{i, j}, h_1^{i, j}, h_2^{i, j}, h_3^{i, j})$ if it is clear from the context.
Then we construct $X_i$ by gluing each $H^{i,j}$ for $j\in [\ell]$ onto $X_i'$ as follows. 
\begin{itemize}
    \item If $C_{j}^i$ contains $x_i$, then identify $g^{i, j}$ with $v_{4j}^i$ and $h^{i, j}$ with $v_{4j-2}^i$.
    \item If $C_{j}^i$ contains $\lnot x_i$, then we identify $g^{i, j}$ with $v_{4j-2}^i$ and $h^{i, j}$ with $v_{4j}^i$.
\end{itemize}
See~\cref{fig:VarSplitter} for an illustration.
\begin{figure}[h]
    \centering
    \includegraphics[scale=0.6]{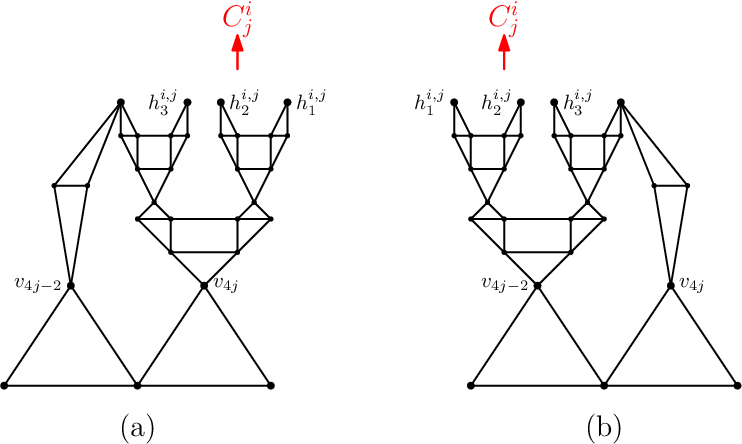}
    \caption{An illustration of two cases for $X_i$. (a) $C^i_j$ contains $x_i$. (b) $C^i_j$ contains $\lnot x_i$.}
    \label{fig:VarSplitter}
\end{figure}

\subparagraph{Clause gadget.}
For a clause $C_i \in \mathcal{C}$, let $x_1^i, x_2^i, x_3^i \in \mathcal{X}$ be its neighbors in $G_\varphi$, whose indices are assigned according to the cyclic order in the planar drawing of $G_\varphi$.  
We first define an auxiliary graph $W$ on $6$ vertices as follows: 
We start with a cycle of length $6$ on the vertex set $\set{a_1,a_2,\ldots, a_6}$ and add three chords  $a_1a_5$, $a_2a_4$, and~$a_2a_5$.

Let $W^1, W^2, W^3$ be the three disjoint copies of $W$ with vertex sets $V(W^j)=\{i \in [6] \mid a_{i}^j\}$ for each $j \in [3]$, where $a_{i}^j$ denotes the copy of $a_i$ in $W^j$.
Then starting from $W^1,W^2$ and $W^3$, we construct a clause gadget $D$ by identifying $a_{1}^j$ with $a_{3}^{j+1}$ for each $j \in [3]$; and adding edges $a_{2}^j a_{2}^{j+1}$ for each $j\in [3]$ where $a_{i}^4=a_{i}^1$ for $i \in [6]$.
See~\cref{fig:GadgetClause} for an illustration.

\begin{figure}[h]
    \centering
    \includegraphics[scale=0.6]{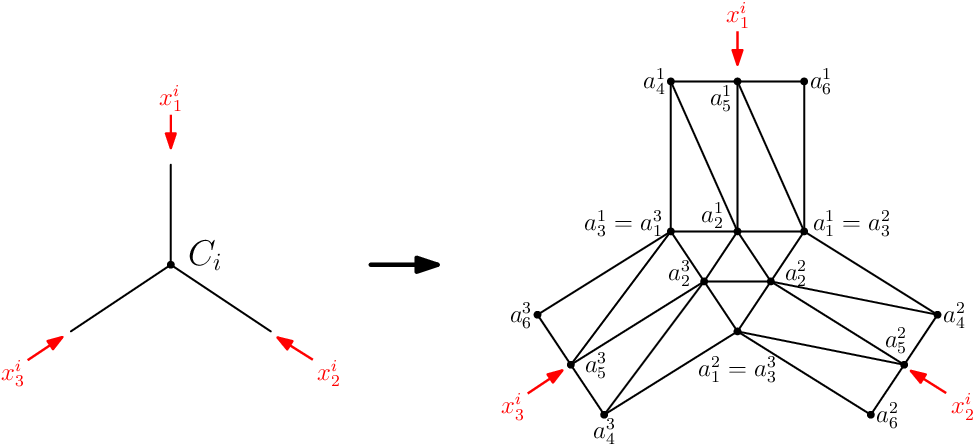}
    \caption{An illustration of a clause gadget.}
    \label{fig:GadgetClause}
\end{figure}

\subparagraph{Putting all together: the graph $G'_\varphi$.}
Finally, we construct $G'_\varphi$ from $G_\varphi$ as follows. See also~\cref{fig:ReductionIllust} for an illustration.
\begin{itemize}
    \item For each $x_i\in\mathcal{X}$, we replace it by $X_i$.
    \item For each $C_i\in \mathcal C$, we replace it by the copy of $D$ on vertex set $\{a_{k}^{i, j} \mid j \in [3], k \in [6]\}$, where each $a_k^{i, j}$ is a copy of $a_k^j$.
    \item For each edge $x_i C_{i'}$ of $G_\varphi$, let $j$ and $j'$ be indices such that $C_j^i = C_{i'}$ and $x_{j'}^{i'} = x_i$. 
    Note that $C_j^i$ and $x_{j'}^{i'}$ are a clause and a variable that are adjacent to $x_i$ and $C_{i'}$ in $G_\varphi$, respectively. 
    Then we identify $h_k^{i, j}$ and $a_{k+3}^{i', j'}$ for each $k \in [3]$ if $x_i \in C_{i'}$ and $h_k^{i, j}$ and $a_{7-k}^{i', j'}$ for each $k \in [3]$ if $\lnot x_i \in C_{i'}$.
\end{itemize}

Note that variable and clause gadgets are planar.
In addition, we can identify $h_k^{i, j}$ and $a_{k+3}^{i', j'}$ or $a_{7-k}^{i', j'}$ along the edge $x_i C_{i'}$ of $G_\varphi$. Since $G_\varphi$ is planar, this identification can be done without making a crossing. 
Thus, $G'_\varphi$ is a planar graph.

In the rest of this section, we show that it is a correct reduction. Before we show~\cref{thm:hardness_K3_factor}, we state key properties of the splitter gadget~(\cref{claim:splitter}), the variable gadget~(\cref{claim:gadget_for_variable}), and the clause gadget~(\cref{claim:gadget_for_clause}).

\begin{lemma}\label{claim:splitter}
    Let $T$ be a subgraph of $H^{\mathrm{sp}}$ consisting of vertex-disjoint triangles such that $V(T) \supseteq V(H^{\mathrm{sp}}) \setminus \{g, h, g_1, g_2, g_3\}$ and $\abs{V(T) \cap \{g, h\}}=1$. Then $V(T)\cap \{g,h,h_1,h_2,h_3\}$ is either $\{g\}$ or $\{h,h_1,h_2,h_3\}$.
    In addition, for each case, there exists such a subgraph $T$.
\end{lemma}
\begin{proof}
    We begin with the following observation.
\begin{observation}\label{claim:small_splitter}
Let $T \subseteq V(L)$ be a subgraph consisting of vertex-disjoint triangles in $L(g,g_1,g_2)$ such that $V(T) \supseteq \{h^1_1, h^2_2, h^1_2, h^2_3\}$.
Then either $T$ covers all vertices of $L(g,g_1,g_2)$, or $T$ covers all but $g,g_1$ and $g_2$. Moreover, for each case, there exists such a subgraph $T$.
\end{observation}
\begin{clproof}
    We first consider when $T$ does not contain $g$. 
    Then, to cover $h_1^1$ and $h_2^1$, two triangles $h_1^1h_1^2h_2^2$ and $h_2^1h_3^2h_4^2$ must be contained in $T$.
    Then it is a maximal collection of vertex-disjoint triangles of $L$. Thus, $T$ cannot cover $g_1$ and $g_2$.
    If $g$ is in $V(T)$, then the triangle $gh_1^1h_2^1$ is included in $T$. 
    Then to cover $h_1^2$ and $h_4^2$, the triangles $h_1^2h_2^1g_1$ and $h_3^2h_4^2g_2$ should be contained in $T$. In this case, all vertices of $L$ are contained in $T$.
\end{clproof}
    Next, we give a proof of~\cref{claim:splitter}.
    We first consider the case when $g \notin V(T)$ and $h \in V(T)$.
    Then by applying \Cref{claim:small_splitter} to $L$, $g_1$ and $g_2$ are isolated vertices in $T[V(L)]$.
    As $g_1$ and $g_2$ are in $V(T)$, they must be contained in a triangle in $T[V(L^1)]$ and $T[V(L^2)]$, respectively.
    By applying \Cref{claim:small_splitter} to $L^1$ and $L^2$, $T[V(L^1) \cup V(L^2)]$ must contain $g_1, g_2, g_3$. 
    In addition, by \Cref{claim:small_splitter}, $T \setminus \{h, w_1, w_2\}$ is uniquely determined and to guarantee $h, w_1, w_2 \in V(T)$, the triangle $hw_1w_2$ is contained in $T$.

    The case when $g \in V(T)$ and $h \notin V(T)$ is similar.
    By \Cref{claim:small_splitter} on $L$, $g_1$ and $g_2$ are contained in a triangle of $T[V(L)]$.
    Then as $T[V(L_1)]$ and $T[V(L_2)]$ contains $g_1$ and $g_2$ as an isolated vertex, by \Cref{claim:small_splitter}, $g_1, g_2, g_3$ are not contained in a triangle of $T[V(L_1) \cup V(L_2)]$ and this shows $g_1, g_2, g_3 \notin V(T)$.
    In addition, $T \setminus \{h,w_1, w_2\}$ is uniquely determined and to cover $w_1$ and $w_2$ and not $h$, the triangle $g_2^2w_1w_2$ must included in $T$. So it is the unique choice of $T$.
\end{proof}

\begin{lemma}\label{claim:gadget_for_variable}
    Each variable gadget $X_i$ contains exactly two maximal triangle tilings that cover all vertices of $V(X_i) \setminus \{h_k^{i, j} \mid j\in [\ell], k \in [3]\}$. Furthermore, one of them covers all $h_k^{i, j}$ with $x_i \in C_j^i$ and does not cover all $h_k^{i, j}$ with $\lnot x_i \in C_j^i$, and the other one covers all $h_k^{i, j}$ with $\lnot x_i \in C_j^i$ and does not cover all $h_k^{i, j}$ with $x_i \in C_j^i$.
\end{lemma}
\begin{proof}
        Let $T \subseteq X_i$ be a subgraph consists with vertex-disjoint triangles and $V(T) \supseteq V(X_i) \setminus \{h_k^{i, j} \mid j\in [\ell], k \in [3]\}$.
    Then as $T$ contains all $v_1^i, \ldots, v_{4\ell-1}^i$, either $T$ contains all the triangles $v_{4k-3}^iv_{4k-2}^iv_{4k-1}^i$ for $k \in [\ell]$ or contains all triangles $v_{4k-1}^iv_{4k}^iv_{4k+1}^i$ for $k \in [\ell]$ where $v_{4\ell+1}^i=v_1^i$.

    We analyze the case in which $T$ contains all triangles $v_{4k-3}^iv_{4k-2}^iv_{4k-1}^i$ for $k \in [\ell]$. 
    The remaining case is symmetric and can be handled analogously. 
    Then for each $j \in [\ell]$, if $x_i \in C_j^i$, then $g^{i, j}$ was identified to $v_{4j}^i$ and $h^{i, j}$ was identified to $v_{4j-2}^i$.
    Thus, $T[V(H^{i, j})]$ contains $h^{i, j}$ as an isolated vertex and $g^{i, j}$ should be contained in a triangle of $T[V(H^{i, j})]$.
    By \Cref{claim:splitter}, triangles in $T[V(H^{i, j})]$ cover neither $h_1^{i, j}$, $h_2^{i, j}$, nor $h_3^{i, j}$. Thus, $h_1^{i, j}, h_2^{i, j}, h_3^{i, j}$ are not contained in $T$.
    If $\lnot x_i \in C_j^i$, then by the same argument, $T[V(H^{i, j})]$ contains $g^{i, j}$ as an isolated vertex and $h^{i, j}$ should be contained in a triangle of $T[V(H^{i, j})]$. Thus, by \Cref{claim:splitter}, triangles in $T[V(H^{i, j})]$ contain all $h_1^{i, j}, h_2^{i, j}, h_3^{i, j}$. Furthermore, by the uniqueness of \Cref{claim:splitter}, it is the unique possible triangle tiling in this case.
\end{proof}

\begin{lemma}\label{claim:gadget_for_clause}
    Let $U_i = \{a_4^i, a_5^i, a_6^i\} \subseteq V(D)$ for $i \in [3]$.
    Then for a set of indices $I \subseteq [3]$, the graph $D - \bigcup_{i \in I} U_i$ has a perfect triangle tiling if and only if $\abs{I}=2$. 
\end{lemma}
\begin{proof}
    We first note that a triangle tiling of $W$ that covers $a_4, a_5, a_6$ must be a perfect triangle tiling. It is because $a_4, a_5, a_6$ do not form a triangle, so a triangle tiling that covers $a_4, a_5, a_6$ should contain at least $6$ vertices.
    Also note that there is a unique such tiling that consists of $a_1a_5a_6$ and $a_2a_3a_4$.

    We now consider when $\abs{I} \leq 1$. By symmetry, we may assume that $1, 2 \notin I$.
    Then by previous observation, to cover $a_4^j, a_5^j, a_6^j$, one must use triangles $a_1^ja_5^ja_6^j$ and $a_2^ja_3^ja_4^j$ for $j \in [2]$. However, as $a_3^1 = a_1^2$, those triangles cannot be contained in a perfect triangle tiling simultaneously. 

    When $\abs{I}=3$, one can check $D- \bigcup_{i=1}^3 U_i$ does not contain two vertex-disjoint triangles. Thus, it does not have a perfect triangle tiling.
    When $\abs{I}=2$, we may assume that $I=\{2, 3\}$.
    Then, there is a perfect triangle tiling consisting of $a_1^1a_5^1a_6^1$, $a_2^1a_2^1a_4^1$, and $a_2^2a_1^2a_2^3$.

\end{proof}

\begin{proof}[Proof of~\cref{thm:hardness_K3_factor}]
We now prove that $\varphi$ is a \textsc{Yes} instance of $1$-in-$3$ SAT if and only if $G'_\varphi$ contains a perfect triangle tiling.
We first assume that $G'_\varphi$ contains a perfect triangle tiling. 
Then let $T \subseteq G$ be any subgraph consisting of vertex-disjoint triangles that covers all the vertices of $G'_\varphi$.
Then triangles in $T[V(X_i)]$ cover all the vertices in $V(X_i) \setminus \{h_k^{i, j} \mid j\in [\ell], k \in [3]\}$.
Thus, by \Cref{claim:gadget_for_variable}, triangles in $T[V(X_i)]$ either cover all $h_k^{i, j}$ with $x_i \in C_j^i$ and do not cover all $h_k^{i, j}$ with $\lnot x_i \in C_j^i$, or cover all $h_k^{i, j}$ with $\lnot x_i \in C_j^i$ and do not cover all $h_k^{i, j}$ with $x_i \in C_j^i$.
For the formal case, we assign \textsc{False} to $x_i$ and for the latter case, we assign \textsc{True} to $x_i$.
Then we claim for each clause, exactly one literal has value \textsc{True} and thus $\varphi$ is a \textsc{Yes} instance of $1$-in-$3$ SAT.
It is because each $h_k^{i, j}$ that is not covered by a triangle in~$T[V(X_i)]$ should be covered by triangles in~$T[D_{i'}]$ for some $i' \in [m]$.
By \Cref{claim:gadget_for_clause}, to $T$ be a triangle factor, exactly one of $\{a_4^{i', 1}, a_5^{i', 1}, a_6^{i', 1}\}$, $\{a_4^{i', 2}, a_5^{i', 2}, a_6^{i', 2}\}$, and $\{a_4^{i', 2}, a_5^{i', 2}, a_6^{i', 2}\}$ are not covered by triangles in $\bigcup_{i=1}^n T[V(X_i)]$. This implies that every clause has exactly one literal with the value \textsc{True}.

To prove the converse, we fix an assignment of $x_i$ such that each clause has exactly one literal with the value \textsc{True}.
We now construct a perfect triangle tiling of $G'_\varphi$ as follows.
First, we choose a triangle tiling on each $X_i$ by choosing one of the triangle tilings in \Cref{claim:gadget_for_variable} where the formal one is chosen when $x_i$ is assigned to \textsc{False}.
Then for each $D_{i'}$, we use a triangle tiling in \Cref{claim:gadget_for_clause}. As exactly one of the $\{a_4^{i', 1}, a_5^{i', 1}, a_6^{i', 1}\}$, $\{a_4^{i', 2}, a_5^{i', 2}, a_6^{i', 2}\}$, and $\{a_4^{i', 2}, a_5^{i', 2}, a_6^{i', 2}\}$ are not covered by the previous tiling, this can be done.
By combining all the triangles, we obtain the desired perfect triangle tiling.
\end{proof}

\section{Concluding remarks}\label{sec:concluding}

Including our results on \textsc{$\mathcal{F}$-subgraph-free Edge Deletion}, in many applications of the product structure theorem, it is beneficial when an explicit embedding into $H \boxtimes P$ is given. 
For instance, one can compute a bounded queue layout of $G$ for a given embedding into $H \boxtimes P$~\cite{dujmovic2020planar, Morin2021AlgoProduct}.
Moreover, even if we can ``approximate'' the product structure, above mentioned applications can still be obtained, albeit they may have weak guarantees.
As a consequence, it is natural to ask whether one can efficiently compute the approximated product structure for a graph $G$ having a product structure. 

\begin{question}\label{ques:compute_prod_struct}
Let $G$ be a subgraph of the product structure
$H\boxtimes P$, where $H$ is a graph of tree-width $w$ and $P$ is a path.
Then, does there exist an FPT-approximation algorithm for finding a product structure $H'\boxtimes P$, where the tree-width of $H'$ is at most $f(w)$ for some computable function $f$?
\end{question}

An affirmative answer to \cref{ques:compute_prod_struct} implies that \textsc{$\mathcal{F}$-subgraph-free Edge Deletion} can be solved in FPT time for every class of graphs with a product structure. 
Furthermore, we would like to ask whether \textsc{$\mathcal{F}$-subgraph-free Edge Deletion} can be solved efficiently without applying the answer of~\cref{ques:compute_prod_struct}.

\begin{question}\label{ques:without_embedding}
Does there exist an FPT algorithm for \textsc{$\mathcal{F}$-subgraph-free Edge Deletion} when given an input graph that has a product structure but an embedding is not given?
\end{question}

\section*{Acknowledgement}
This work was initiated at the 2024 Korean Student Combinatorics Workshop~(KSCW2024) at Gongju, South Korea. The authors would like to thank IBS Discrete Mathematics Group for supporting the workshop. The authors are also grateful for O-joung Kwon, Sang-il Oum, and Sebastian Wiederrecht for valuable comments that improved the presentation of the paper.

\bibliographystyle{plain}
\bibliography{FFEDSTACS/EdgeDeletion}

\end{document}